\documentclass{article}
\usepackage[applemac]{inputenc}
\usepackage[english]{babel}

\usepackage{latexsym}
\usepackage{amssymb}
\usepackage{amsmath}
\usepackage{amsthm}
\usepackage{caption}
\usepackage{ifthen}
\usepackage{xspace}
\usepackage{tikz}
\usepackage{enumerate}

\usetikzlibrary{%
  arrows,%
  automata,%
  positioning,%
  decorations.pathmorphing,%
  decorations.pathreplacing,%
  shadows} 

\newcommand{\set}[2]{\left\{#1\mathrel{\left|\vphantom{#1}\vphantom{#2}\right.}#2\right\}}
\newcommand{\oneset}[1]{\left\{\mathinner{#1}\right\}}

\newcommand{\abs}[1]{\left|\mathinner{#1}\right|}

\newcommand{\N}{\mathbb{N}}

\newcommand{\Rp}{\mathbb{R}^{\ge 0}}

\newcommand{\Oh}{\mathcal{O}}
\newcommand{\cA}{\ensuremath{\mathcal{A}}\xspace}

\newcommand{\cF}{\ensuremath{\mathcal{F}}\xspace}

\newcommand{\cH}{\ensuremath{\mathcal{H}}\xspace}
\newcommand{\cI}{\ensuremath{\mathcal{I}}\xspace}

\newcommand{\cL}{\ensuremath{\mathcal{L}}\xspace}

\newcommand{\cQ}{\ensuremath{\mathcal{Q}}\xspace}

\newcommand{\cS}{\ensuremath{\mathcal{S}}\xspace}

\newcommand{\bA}{\ensuremath{\mathsf{A}}\xspace}
\newcommand{\bC}{\ensuremath{\mathsf{C}}\xspace}
\newcommand{\bG}{\ensuremath{\mathsf{G}}\xspace}
\newcommand{\bT}{\ensuremath{\mathsf{T}}\xspace}

\newcommand{\BAR}{\overline{\phantom{ii}}}

\newcommand{\smalloverline}[1]
{{\mspace{1mu}\overline{\mspace{-1mu}#1\mspace{-1mu}}\mspace{1mu}}}
\newcommand{\ov}[1]{\smalloverline{#1}\vphantom{#1}}
\newcommand{\ovc}[1]{\smalloverline{#1}} % captions dislike the vphantom: use \ovc instead

\newcommand{\PSPACE}{\ensuremath{\mathrm{PSPACE}}\xspace}

\newcommand{\NL}{\ensuremath{\mathrm{NL}}\xspace}

\newcommand{\e}{1}

\newcommand{\IFF}{if and only if\xspace}
\newcommand{\hpc}{hairpin completion\xspace}
\newcommand{\hpcs}{hairpin completions\xspace}

\newcommand{\Grr}{Growth indicator\xspace}
\newcommand{\Grrs}{\Grr{}s\xspace}
\newcommand{\grr}{growth indicator\xspace}
\newcommand{\grrs}{\grr{}s\xspace}

\newcommand{\ie}{i.\,e.,\xspace}
\newcommand{\eg}{e.\,g.,\xspace}
\newcommand{\resp}{respectively,\xspace}

\newcommand{\Pref}{\mathrm{Pref}}

\newcommand{\kap}{\kappa} % maybe replace it with 'k' again
\newcommand\Hk{\cH_\kap}

\newcommand\ccH{\ensuremath{\Hk(L_1,L_2)}\xspace}

\newcommand{\sse}{\subseteq}
\newcommand{\es}{\emptyset}
\newcommand{\sm}{\setminus}

\renewcommand{\phi}{\varphi}

\newcommand{\alp}{\alpha}
\newcommand{\bet}{\beta}
\newcommand{\gam}{\gamma}
\newcommand{\del}{\delta}
\newcommand{\lam}{\lambda}
\newcommand{\sig}{\sigma}

\newcommand{\Sig}{\Sigma}

\newcommand{\aba}{{\alpha\beta\ov\alpha}}

\newcommand{\gabag}{\gamma\alpha\beta\ov\alpha\ov\gamma}
\newcommand{\gaba}{\gamma\alpha\beta\ov\alpha}
\newcommand{\abag}{\alpha\beta\ov\alpha\ov\gamma}

\newcommand\RAS[1]{\overset{#1}\Longrightarrow}
\newcommand\ras[1]{\overset{#1}\longrightarrow}
\newcommand\derive{\underset{G}{\Longrightarrow}}
\newcommand\derives{\underset{G}{\overset{*}\Longrightarrow}}
\newcommand\production{\to}

\newcommand\ftt{\mbox{$5'$-to-$3'$}\xspace}
\newcommand\ttf{\mbox{$3'$-to-$5'$}\xspace}

\newcommand{\dead}{t}%{\emptyset}

\newcommand{\bridge}{B}
\newcommand{\mgp}{R}
\newcommand{\CP}{\bridge_{\mu}^{\mgp_\mu}}
\newcommand{\brgr}{\sigma}
\newcommand{\mgpgr}{\rho}

\theoremstyle{plain}
\newtheorem{theorem}{Theorem}[section]
\newtheorem{proposition}[theorem]{Proposition}
\newtheorem{lemma}[theorem]{Lemma}

\theoremstyle{definition}

\newtheorem{example}[theorem]{Example}

\theoremstyle{remark}
\newtheorem{remark}[theorem]{Remark}

\newenvironment{test}[1]
{\begin{trivlist}\item[\hskip\labelsep {\bfseries Test #1:\,}]}
{\end{trivlist}}

\newcommand{\refthm}[1]{Thm.~\ref{#1}}

\newcommand{\reflem}[1]{Lem.~\ref{#1}}
\newcommand{\refprop}[1]{Prop.~\ref{#1}}
\newcommand{\refrem}[1]{Rem.~\ref{#1}}

\newcommand{\reffig}[1]{Fig.~\ref{#1}}
\newcommand{\refsec}[1]{Sect.~\ref{#1}}

\newenvironment{vd}{\noindent\color{blue} VD }{}

\newenvironment{sk}{\noindent\color{red} SK }{}

\newenvironment{vm}{\noindent\color{green} VM }{}

\begin{document}

%%%%%%%%%%%%%%%%%%%%%%%%%%%%%%%%%%%%%%%%%%%%%%%%%%%%%%%%%%%%%%%%%%%%%%%%%%%%%%%%
% title page

\title{Deciding Regularity of Hairpin Completions of Regular Languages
	in Polynomial Time}
\author{Volker Diekert \and Steffen Kopecki \and Victor Mitrana}
%\ead{diekert@fmi.uni-stuttgart.de}
%\ead{kopecki@fmi.uni-stuttgart.de}
%\ead{mitrana@fmi.unibuc.ro}
\date{}
%\address[fmi]
%	{University of Stuttgart,
%	Institute for Formal Methods in Computer Science, \\
%	Universit\"atsstra\ss e 38,
%	70569 Stuttgart, Germany}

%\address[unibuc]
%	{University of Bucharest,
%	Faculty of Mathematics and Computer Science, \\
%	Str.~Academiei~14, 010014, Bucharest, Romania}

\maketitle

\begin{abstract}
	The hairpin completion is an operation on formal languages that has been
	inspired by the hairpin formation in DNA biochemistry and by DNA computing.
	In this paper we investigate the hairpin completion of regular languages.
	
	It is well known that hairpin completions of regular languages are linear
	context-free and not necessarily regular.
	As regularity of a (linear) context-free language is not decidable,
	the question arose whether regularity of a hairpin completion of regular
	languages is decidable.
	We prove that this problem is decidable and we provide a polynomial
	time algorithm.

	Furthermore, we prove that the hairpin completion of regular languages is an
	unambiguous linear context-free language and, as such, it has an
	effectively computable growth function.
	Moreover, we show that the growth of the hairpin completion is exponential
	if and only if the growth of the underlying languages is exponential
	and, in case the hairpin completion is regular, then the hairpin completion
	and the underlying languages have the same growth indicator.	
	
	\quad\par
	
\end{abstract}

%%%%%%%%%%%%%%%%%%%%%%%%%%%%%%%%%%%%%%%%%%%%%%%%%%%%%%%%%%%%%%%%%%%%%%%%%%%%%%%%
%%%%%%%%%%%%%%%%%%%%%%%%%%%%%%%%%%%%%%%%%%%%%%%%%%%%%%%%%%%%%%%%%%%%%%%%%%%%%%%%
\section{Introduction}
A {\em DNA strand} can be seen as a word over the four-letter alphabet
$\oneset{\bA,\bC,\bG,\bT}$ where the letters represent the nucleobases
Adenine, Cytosine, Guanine, and Thymine, respectively.
By {\em Watson-Crick base pairing} two strands may bond to each other if
they have opposite orientation and their bases are pairwise complementary,
where $\bA$ is complementary to $\bT$ and $\bC$ to $\bG$;
see \reffig{fig:bond} for a graphic example.
Throughout the paper we use the bar-notation for the Watson-Crick complement and
its language theoretic pendant, \ie $\ov\bA = \bT$ and $\ov\bC = \bG$.
For base sequences (or words) we let $\ov{a_1\cdots a_m} = \ov{a_m}\cdots\ov{a_1}$;
thus, $\BAR$ is an antimorphic involution.

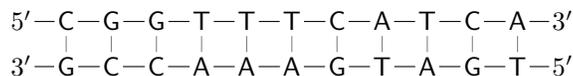
\begin{figure}[ht]
	\center
	\begin{tikzpicture}[text height=1.75ex,text depth=.25ex, node distance=.6cm, inner sep=1pt]
		\node (a1) {$5'$};
		\node (a2) [right of=a1] {$\bC$};
		\node (a3) [right of=a2] {$\bG$};
		\node (a4) [right of=a3] {$\bG$};
		\node (a5) [right of=a4] {$\bT$};
		\node (a6) [right of=a5] {$\bT$};
		\node (a7) [right of=a6] {$\bT$};
		\node (a8) [right of=a7] {$\bC$};
		\node (a9) [right of=a8] {$\bA$};
		\node (a10) [right of=a9] {$\bT$};
		\node (a11) [right of=a10] {$\bC$};
		\node (a12) [right of=a11] {$\bA$};
		\node (a13) [right of=a12] {$3'$};
		
		\draw (a1) -- (a2) -- (a3) -- (a4) -- (a5) -- (a6) -- (a7)
			-- (a8) -- (a9) -- (a10) -- (a11) -- (a12) -- (a13) ;

		\node (b1) [below of=a1,node distance=.6cm] {$3'$};
		\node (b2) [right of=b1] {$\bG$};
		\node (b3) [right of=b2] {$\bC$};
		\node (b4) [right of=b3] {$\bC$};
		\node (b5) [right of=b4] {$\bA$};
		\node (b6) [right of=b5] {$\bA$};
		\node (b7) [right of=b6] {$\bA$};
		\node (b8) [right of=b7] {$\bG$};
		\node (b9) [right of=b8] {$\bT$};
		\node (b10) [right of=b9] {$\bA$};
		\node (b11) [right of=b10] {$\bG$};
		\node (b12) [right of=b11] {$\bT$};
		\node (b13) [right of=b12] {$5'$};
		
		\draw (b1) -- (b2) -- (b3) -- (b4) -- (b5) -- (b6) -- (b7)
			-- (b8) -- (b9) -- (b10) -- (b11) -- (b12) -- (b13);
		
		\foreach \x in {2,...,12}
			\draw [gray] (a\x) -- (b\x);
		
		\begin{scope} [node distance=.3cm]
			\node (c1) [above of=a1] {};
			\node (c2) [above of=a10] {};
			\node (d1) [below of=b10] {};
			\node (d2) [below of=b1] {};
		\end{scope}
		
	\end{tikzpicture}
	\caption{Bonding of two strands: The strands are base-wise complementary and
		the first strand has \ftt orientation whereas the second
		strand has \ttf orientation.}\label{fig:bond}
\end{figure}

The {\em polymerase chain reaction} (PCR) is an technique which is often used in
DNA computing to amplify a {\em template strand} or a fragment of the template strand.
Short DNA sequences, so-called {\em primers}, bond to a part of the template
and thusly select where the {\em extension}, the process where template is complemented,
will start.

The {\em hairpin completion} of a strand can naturally develop during the PCR.
Suppose a strand $\sig$ can be written as $\sig = \gaba$.
Therefore, its suffix $\ov\alp$ can act as a primer to the strand
and form an intramolecular base-pairing which is known as {\em hairpin formation}.
After the extension process we obtain a new strand $\gabag$ which we call a
hairpin completion of $\sig$; see \reffig{fig:hairpin}.
Referring to \cite{Williams1990}, $\alp$ should consist of at least $9$ bases,
otherwise the bond between $\alp$ and $\ov\alp$ is too weak.

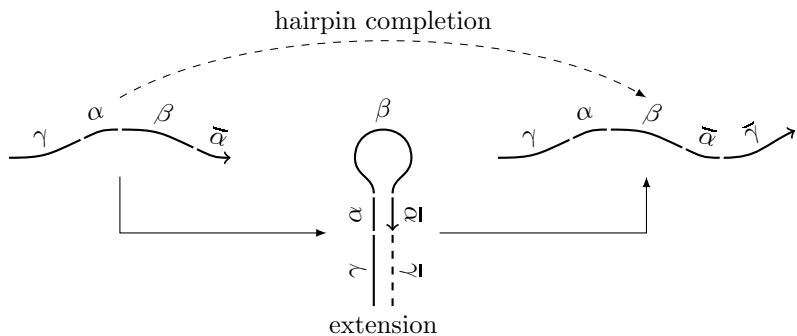
\begin{figure}[ht]
	\centering
	\begin{tikzpicture}[text height=1.5ex,text depth=.25ex]
		\def\ang{25}
		
		\begin{scope}[above,sloped,every path/.style={thick,shorten <=.75pt,shorten >=.75pt}]
			\draw (-1.5,0) .. controls +(0:.5) and +(180+\ang:.5) .. node {$\gam$} (-.5,.25);
			\draw (-.5,.25) .. controls +(\ang:.25) and +(180:.25) .. node {$\alp$} (0,.375);
			\draw (0,.375) .. controls +(0:.5) and +(180-\ang:.5) .. node {$\bet$} (1,.125);
			\draw [->] (1,.125) .. controls +(-\ang:.25) and +(180:.25) 
				.. node {$\ov\alp$} (1.5,0);
		\end{scope}
		
		\draw [-latex] (0,-.25) -- (0,-1) -- (2.75cm,-1);
		
		\begin{scope}[xshift=3.5cm, yshift=-1cm, every node/.style={above, sloped},
				every path/.style={thick,shorten <=.75pt,shorten >=.75pt},above,sloped]
			\draw (-.125,-1) -- node {$\gam$} (-.125,0);
			\draw (-.125,0) -- node {$\alp$} (-.125,.5);
			\draw (-.125,.5) .. controls +(90:.25) and +(-90:.25) .. (-.375,1)
				.. controls +(90:.5) and +(90:.5) .. node {$\bet$} (.375,1)
				.. controls +(-90:.25) and +(90:.25) .. (.125,.5);
			\draw [->] (.125,.5) -- node {$\ov\alp$} (.125,0);
			
			\draw [dashed] (.125,0) -- node {$\ov\gam$} (.125,-1);
			\node at (0,-1.5) {extension};
		\end{scope}

		\draw [-latex] (4.25,-1) -- (7,-1) -- (7,-.25);

		\begin{scope}[xshift=7cm -.5cm,above,sloped,
				every path/.style={thick,shorten <=.75pt,shorten >=.75pt}]
			\draw (-1.5,0) .. controls +(0:.5) and +(180+\ang:.5) .. node {$\gam$} (-.5,.25);
			\draw (-.5,.25) .. controls +(\ang:.25) and +(180:.25) .. node {$\alp$} (0,.375);
			\draw (0,.375) .. controls +(0:.5) and +(180-\ang:.5) .. node {$\bet$} (1,.125);
			\draw (1,.125) .. controls +(-\ang:.25) and +(180:.25) 
				.. node {$\ov\alp$} (1.5,0);
			\draw [->] (1.5,0) .. controls +(0:.5) and +(180+\ang:.5) 
				.. node {$\ov\gam$} (2.5,.375);
		\end{scope}
		
			\draw [-latex,dashed] (0,.8) .. controls +(30:2) and +(150:2) ..
				node [above] {hairpin completion} (7,.8);
	\end{tikzpicture}
	
	\caption{Hairpin completion of a strand or a word.}
	\label{fig:hairpin}
\end{figure}

Hairpin completions are often seen as undesirable byproducts that occur during
DNA computations and, therefore, sets of DNA strands have been investigated
that do not tend to form hairpins or other undesired hybridizations, see \eg
% bad hybridization
\cite{garzon1,garzon2,garzon3,KariKLST05,KariMT07}
and the references within.
On the other hand, DNA algorithms have been designed that make good use of
hairpins and hairpin completions.
For example, the {\em whiplash PCR} is a technique where a single DNA strand computes
one run of a non-deterministic GOTO-machine by repetitive hairpin completions,
where the length of the extended part is controlled by {\em stopper sequences}.
Starting with a huge set of strands, all runs of such a machine can be computed in parallel.
Whiplash PCR can be used to solve NP-complete problems like the Hamiltonian path problem
% whiplash pcr
\cite{Hagiya97,Sakamoto98,Winfree98}.

Motivated by the hairpin formation in biochemistry,
the hairpin completion of formal languages has been introduced in 2006 by
Cheptea, Mart\'\i{}n-Vide, and Mitrana \cite{ChepteaMM06}.
This paper continues the investigation of hairpin formation 
{}from a purely formal language theoretical viewpoint.
The hairpin completion of languages $L_1$ and $L_2$ contains all right hairpin completions
(as in \reffig{fig:hairpin}) of all words in $L_1$
and all left hairpin completions (a word $\abag$ is extended to the left by $\gam$)
of all words in $L_2$.
A formal definition of this operation is given in \refsec{sec:def:hpc}.
The hairpin completion and some related operations have been investigated in a series of papers,
see \eg 
% hairpin papers
\cite{Ito2010,Kopecki11,ManeaMM09,ManeaMM10,DBLP:conf/cie/ManeaM07,ManeaMY09tcs,ManeaMY10}.

It is known from \cite{ChepteaMM06} that the hairpin completion of regular languages
is not necessarily regular but it is always linear context-free.
As regularity of a linear context-free language (given as grammar) is undecidable,
the question arose if regularity of the hairpin completion of regular languages can
be decided.
This question was first posed in 2006 \cite{ChepteaMM06}.
We answered this question positively at \mbox{ICTAC~2009} \cite{DiekertKM09} when
we proved that the problem is decidable in polynomial time.
In this first approach we were not precise about the degree of the polynomial;
it was about~20.
In a later approach, which was presented at \mbox{CIAA~2010} \cite{DieKop10},
we improved the decision algorithm and provided,
that the problem is solvable in $\Oh(n^8)$, where $n$ bounds the size of the two
input DFAs (deterministic finite automata), accepting $L_1$ and $\ov{L_2}$, respectively.
Furthermore, for $L_2 = \es$ we provided a time complexity of $\Oh(n^2)$ and
for $L_1 = \ov{L_2}$ we provided $\Oh(n^6)$.
In the second paper we also showed that the problem is \NL-complete
(\NL is the class of problems that are solvable by a non-deterministic
algorithm using logarithmic space), in particular,
the problem is contained in {\em Nick's Class}
which means it is efficiently solvable in parallel, see \eg \cite{pap94}.
Moreover, we proved that the hairpin completion of regular languages has
an unambiguous linear representation.
Thus, its generating function is an effectively computable rational function.

This paper is organized as follows.
In \refsec{sec:pn}, we formally define the hairpin completion operation,
we lay down our notation, and 
we briefly introduce the concepts of formal language theory that we will use later.
Then, we start our investigation of hairpin completions of regular languages,
in \refsec{sec:unambiguity}, by providing an unambiguous linear grammar generating
the hairpin completion of two given regular languages.
\refsec{sec:algorithm} is devoted to the polynomial time algorithm
that decides the regularity of the hairpin completion of regular languages.
In the final chapter, \refsec{sec:growth}, we discuss the relation of the
growth of the hairpin completion with the growths of the underlying regular
languages.

This paper is the journal version of results that have been presented at
ICTAC~2009. It uses the improvements which were presented at CIAA~2010
and it contains some additional results.

%%%%%%%%%%%%%%%%%%%%%%%%%%%%%%%%%%%%%%%%%%%%%%%%%%%%%%%%%%%%%%%%%%%%%%%%%%%%%%%
%%%%%%%%%%%%%%%%%%%%%%%%%%%%%%%%%%%%%%%%%%%%%%%%%%%%%%%%%%%%%%%%%%%%%%%%%%%%%%%
\section{Preliminaries and Notation}\label{sec:pn}

We assume the reader to be familiar with the fundamental concepts of formal
language theory and automata theory, see \cite{HU}. 
 
By $\Sig$ we denote a finite alphabet with at least two letters which is equipped
with an \emph{involution}  $\BAR\colon \Sig\to \Sig$.
An involution for a set is  a bijection  such that
$\overline{\ov{a}} = a$ for all $a \in \Sig$. 
(In a biological setting we may think of $\Sig = \{A,C,G,T\}$ with $\ov A= T$ and $\ov C= G$.)
We extend this involution to words $a_1 \cdots a_n $ 
by  $\ov{a_1 \cdots a_n} =  \ov{a_n} \cdots \ov{a_1}$.
(Just like taking inverses in groups.)
For languages $\ov L$ denotes the set $\set{\ov w}{w\in L }$.
The set of words over $\Sig$ is denoted $\Sig^*$; and the \emph{empty
word} is denoted by $\e$.
By $\Sig^{\leq m}$ we mean the set 
of all words with length at most $m$.

Given a word $w$, we denote by $\abs w$ its length, by $w[i]\in \Sig$ its $i$-th letter,
and by $w[i,j]$ we mean $w[i]w[i+1]\cdots w[j]$.
If $w=xyz$ for some $x,y,z\in \Sig^*$, then $x$ and $z$ are called \emph{prefix} and %, factor, 
\emph{suffix}, respectively.
A prefix or suffix $x$ of $w$ is said to be {\em proper} if $x\neq w$.
The (proper) prefix relation between words $x$ and $w$ is denoted by $x\leq w$
(\resp~$x < w$).

%%%%%%%%%%%%%%%%%%%%%%%%%%%%%%%%%%%%%%%%%%%%%%%%%%%%%%%%%%%%%%%%%%%%%%%%%%%%%%%%
\subsection{Haiprin completion}\label{sec:def:hpc}
Let $L_1$ and $L_2$ be languages in $\Sig^*$.
By $\kap$ we denote a (small) constant that gives a lower bound for the length of primers.
%, in applications we might think of $\kap = 9$. 
We define the \emph{hairpin completion} $\ccH$ by 
\begin{equation*}
	\ccH= \set{\gamma \alpha \beta \ov{\alpha} \ov{\gamma}}{(\gamma \alpha \beta \ov{\alpha}\in L_1
		\vee \alpha \beta \ov{\alpha}
		\ov{\gamma} \in L_2 ) \wedge \abs\alp \ge \kap}.
\end{equation*}
Three cases are of main interest:
\begin{enumerate}[\quad 1.)]
	\item $L_1=L_2$,
	\item $L_1= \ov{L_2}$, and
	\item $L_1= \es$ or $L_2=\es$.
\end{enumerate}
Compared to the definition of the \hpc in \cite{ChepteaMM06,ManeaMY09tcs}
case~1 corresponds to the the two-sided \hpc and case~3 to the one-sided \hpc.
In many biochemical applications a strand and its complement
always co-occur, thus, the assumption $L_1 = \ov{L_1} = \ov{L_2}$ is natural, too,
and it is a covered by case~2.

%%%%%%%%%%%%%%%%%%%%%%%%%%%%%%%%%%%%%%%%%%%%%%%%%%%%%%%%%%%%%%%%%%%%%%%%%%%%%%%%
\subsection{Linear Context-free Grammars and Unambiguity}

A grammar $G$ is a tuple $G = (V,\Sig,P,\cS)$ where $V$ is the finite set of {\em non-terminals},
$\Sig$ is the alphabet (the set of {\em terminals}), $P$ is the finite set of
{\em production rules}, and $\cS\sse V$ is the set of {\em axioms}.
(Note that we allow a set of axioms rather than the more usual 
restriction to have exactly one axiom $S$.)
A grammar is called {\em context-free}, if every rule in $P$ is of the form $A \production w$
where $A\in V$ and $w\in(V\cup\Sig)^*$;
a grammar is called {\em linear context-free}, or simply {\em linear}, if,
in addition, $w$ contains at most one non-terminal.
For a context-free grammar $G$, 
a {\em derivation step} is denoted by $uAv \derive uwv$,
where $A\production w$ is a production rule in $P$ and $u,v\in(V\cup \Sig)^*$.
By~$\derives$, we denote the reflexive and transitive closure of $\derive$
and we call $u\derives v$ (with $u,v\in(V,\Sig)^*$) a {\em derivation}.
The language generated by $G$ is the set of terminal words
\begin{equation*}
	L(G) = \set{w\in\Sig^*}{\exists A\in \cS \colon A\derives w}.
\end{equation*}

A linear grammar $G$ is said to be {\em unambiguous} if for every
word $w\in L(G)$, there is exactly one derivation $A \derives w$
where $A\in \cS$;
in particular, there is only one axiom $A$ that derivates $w$.
(For general context-free grammars we would require that there is exactly one
{\em left-most derivation $A\derives w$};
but in case of linear grammars, these definitions coincide.)

A language $L$ is called {\em (unambiguous) linear} if it is generated 
by an (unambiguous) linear grammar.

%%%%%%%%%%%%%%%%%%%%%%%%%%%%%%%%%%%%%%%%%%%%%%%%%%%%%%%%%%%%%%%%%%%%%%%%%%%%%%%%
\subsection{Generating Functions}\label{sec:grr}
For a profound discussion of formal power series and how the growth of regular and
unambiguous linear languages can be calculated we refer to
\cite{BerstelReutenauer2010,CeccheriniSilberstein2005,GawrychowskiKRS08,MR42:4343}.
We content ourselves with a few basic facts.
The \emph{growth} or \emph{generating function} $g_L$ 
of a formal language $L$ is defined as
\begin{equation*}
	g_L(z)= \sum_{m \geq 0}\abs{L \cap \Sig^{ m}} z^m.
\end{equation*}
We can view $g_L$ as a formal power series or as an analytic 
function in one complex variable where the radius of convergence is strictly 
positive. The radius of convergence is at least $1\slash\abs \Sig$. 

It is well-known that the growth of a regular language $L$ is effectively rational,
\ie it is a quotient of two polynomials, which can be effectively
calculated. 
The same is true for unambiguous linear languages as soon as we know a
 generating unambiguous linear grammar.
In particular, the growth is either polynomial or exponential.
If the growth is exponential, then
there exists an algebraic number $\lam_L \in \Rp$,
its \grr, such that $\abs{L \cap \Sig^{m}}$
behaves essentially as $\lam_L^m$. More precisely, 
for a language $L$, its {\em \grr} is defined
as the non-negative real number
$\lam_L$ where
\begin{equation*}
	\lam_L = \inf\set{\lam\in \Rp}
		{\exists c>0, \forall m\in\N \colon \abs{L\cap \Sig^{m}}\le c \lam^m}.
\end{equation*}
The growth of a language $L$ is
\begin{enumerate}[\quad 1.)]
	\item exponential if $1 < \lam_L \le \abs\Sig$,
	\item sub-exponential but infinite if $\lam_L = 1$, and
	\item finite if $\lam_L = 0$.
\end{enumerate}
Note that other values for $\lam_L$ do not occur and
that $\lam_L$ is the inverse of the convergence radius of $g_L(z)$.
As we discussed above, the growth of an unambiguous linear language $L$ is 
either polynomial or exponential; thus, if $\lam_L = 1$, the growth of
$L$ can be considered polynomial.
Note that  regular languages of polynomial growth have a very restricted form:
It is well known that a regular language has polynomial growth \IFF
it can be written as a finite union of languages of the 
form $u_0u_1^*u_2\cdots u_{2k-1}^*u_{2k}$ where $u_i$ are words, see \eg \cite{Szilard1992}.
Thus, the more interesting situation occurs when a language has exponential growths.
It is then when the \grr becomes significant. 
 
%%%%%%%%%%%%%%%%%%%%%%%%%%%%%%%%%%%%%%%%%%%%%%%%%%%%%%%%%%%%%%%%%%%%%%%%%%%%%%%%
\subsection{Regular Languages and Finite Automata}

Regular languages can be specified by non-deterministic finite automata (NFA)
$\cA= (\cQ, \Sig, E, \cI, \cF)$,
where $\cQ$ is the finite set of \emph{states}, $\cI \sse \cQ$ is the set of 
\emph{initial states}, and  $\cF \sse \cQ$ is the set of 
\emph{final states}. The set $E$ contains labeled  \emph{transitions} (or  \emph{arcs}),
it is a subset of $\cQ \times \Sig \times  \cQ$.
For a word $w \in \Sig^*$ we write $p \ras{w}q$, if there is a path
from state $p$ to $q$ which is labeled by $w$. Thus,  the
accepted language becomes
\begin{equation*}
	L(\cA) = \set{w\in \Sig^*}{\exists p \in \cI,\, \exists q\in \cF:\, p \ras{w}q}.
\end{equation*}

Later it will be crucial to use also paths which avoid final states. For this 
we introduce a special notation. First remove all arcs $(p,a,q)$ where 
$q\in \cF$ is a final state. Thus, final states do not have incoming arcs anymore
in this reduced automaton. Let us write 
$p \RAS{w}q$, if there is a path in this reduced automaton 
from state $p$ to $q$ which is labeled by the word $w$. Note that for such a 
path $p \RAS{w}q$ we allow $p\in \cF$, but on the path we never meet any final state again.

An NFA is called a deterministic finite automaton (DFA), if it has one initial state
and for every state $p  \in \cQ$ and every letter $a \in \Sig$ there is exactly one arc $(p,a,q) \in E$.
In particular,
a DFA in this paper is always complete, thus we can read every word
to its end.  We also write $p \cdot w = q$, if 
$p \ras{w}{}q$. This yields a (totally defined) function $\cQ \times \Sigma^* \to
\cQ$, which  defines an action of $\Sigma^*$ on $\cQ$
on the right.

\subsection{Notation}
Throughout the paper, $L_1$ and $L_2$ denote fixed  regular languages
in $\Sig^*$.
We use  a DFA accepting $L_1$ as well as a DFA accepting  $L_2$,
%but the DFA for $L_2$ has to work 
which works from right-to-left.
However, instead of introducing this concept we use a %minimal 
DFA (working as usual from left-to-right), which 
accepts  $\ov {L_2}$. This  automaton has the same number of states
(and is structurally isomorphic to) as a DFA accepting the \emph{reversal language} of $L_2$.
Our input is therefore given by two  DFAs $\cA_i = (\cQ_i,\Sig,E_i,\{q_{0i}\},\cF_i)$ for $i = 1,2$ which 
accept the languages $L_1$ and $\ov{L_2}$, respectively.
We let $n_1 = \abs{\cQ_1}$, $n_2 = \abs{\cQ_2}$, and
we let $n = \max\oneset{n_1, n_2}$ be the input size.

%%%%%%%%%%%%%%%%%%%%%%%%%%%%%%%%%%%%%%%%%%%%%%%%%%%%%%%%%%%%%%%%%%%%%%%%%%%%%%%%
%%%%%%%%%%%%%%%%%%%%%%%%%%%%%%%%%%%%%%%%%%%%%%%%%%%%%%%%%%%%%%%%%%%%%%%%%%%%%%%%
\section{Unambiguity of $\ccH$}
\label{sec:unambiguity}

In this section we prove that the hairpin completion $\ccH$
is an unambiguous linear context-free language.
The result is not needed for deciding
regularity of $\ccH$, but it came out as a byproduct of the
decision procedure. However, the result turned out to be rather 
fundamental for the understanding of \hpcs of regular languages, in general. 
In particular, it allows to compute the growths of   $\ccH$ and to compare it 
with the growths of the languages $L_1$ and $L_2$, see \refsec{sec:growth}. 
Moreover,  ideas of  this section, will be reused
when we provide the algorithm deciding the regularity of \ccH.
Therefore we begin with the following result.

\begin{theorem}The \hpc is unambiguous linear context-free.
	Moreover,  there is an effective construction of a generating
	unambiguous linear grammar $G$ for  \ccH such that 
	the size of the grammar $G$ is in $\Oh(n_1^2n_2^2)\sse \Oh(n^4)$.
\end{theorem}

\begin{proof}
The  basic observation is that 
every word $\pi\in\ccH$ has a unique factorization $\pi = \gabag$ such that
\begin{enumerate}[\quad 1.)]
	\item $\gaba\in L_1$ or $\abag\in L_2$,
	\item $\abs \alp = \kap$,
	\item if a prefix of $\pi$ belongs to $L_1$, then it is a prefix of $\gaba$, and
	\item if a suffix of $\pi$ belongs to $L_2$, then it is a suffix of $\abag$.
\end{enumerate}
In other words, among all factorizations which satisfy the first condition
and where $\abs\alp \ge \kap$,
we choose the factorization where $\abs \alp = \kap$ and the length of $\gam$ is  minimal.
In such a factorization we call $\gam\alp\le \pi$ the
{\em minimal gamma-alpha-prefix} of $\pi$.
This factorization yields runs in the DFAs $\cA_1$ and $\cA_2$ as in \reffig{fig:urun}.
(Recall that $\cA_2$ accepts $\ov{L_2}$ and $\ov\pi = \gam \alp \ov\bet\ov\alp\ov\gam$.)
As $\pi$ determines the factors $\gam$ and $\alp$, the states
$c_i$, $d_i$, $e_i$, $f_i$, and $q_i'$ (for $i=1,2$) are determined by $\pi$ as well.

\begin{figure}[ht]
	\vspace{-.5\baselineskip}
	\begin{align*}
		\cA_1: \quad& q_{01} \ras{\gam} c_1 \ras{\alp} d_1 \ras {\bet }{}
			e_1 \ras {\ov \alp}{} f_1 \RAS {\ov \gam }{} q_1' \\
		\cA_2: \quad& q_{02} \ras{\gam} c_2 \ras{\alp} d_2\ras {\ov \bet }{}
			e_2 \ras {\ov \alp}{} f_2 \RAS {\ov \gam}{} q_2'
	\end{align*}
	\caption{The runs defined by $\pi\in \ccH$ where $\gam\alp$ is
		the minimal gamma-alpha-prefix
		and, therefore, $f_1\in \cF_1$ or $f_2\in \cF_2$.}
	\label{fig:urun}
\end{figure}

Vice versa, every path of this form (where $\abs\alp = \kap$) defines one word $\pi=\gabag$
from the hairpin completion \ccH such that $\gam\alp$ is its minimal gamma-alpha-prefix.

By this observation, we can use quadruples of states 
in order to define the unambiguous linear grammar $G$ that generates
the hairpin completion \ccH.
For every $(p_1,p_2,q_1,q_2) \in \cQ_1 \times \cQ_2\times \cQ_1\times \cQ_2$
we define a regular language
\begin{equation*}
	\bridge(p_1,p_2,q_1,q_2) = \set{w\in \Sig^*}{p_1\cdot w = q_1 \wedge p_2\cdot \ov w = q_2}.
\end{equation*}
 Thus, in \reffig{fig:urun} we have 
$\pi \in \bridge(q_{01},q_{02},q_1',q_2')$, $\aba\in \bridge(c_1,c_2,f_1,f_2)$, and
$\bet\in \bridge(d_1,d_2,e_1,e_2)$. Later a tuple $(p_1,p_2,q_1,q_2)$ with 
$\bridge(p_1,p_2,q_1,q_2)\neq \es$ is called a \emph{basic bridge}, see \refsec{sec:automaton}.

Furthermore, in our grammar $G$, we let $\bridge(p_1,p_2,q_1,q_2)$ be the non-terminal
that derives all words from the language $\bridge(p_1,p_2,q_1,q_2)$.
Compared to \reffig{fig:urun}, we intend that
$\bridge(d_1,d_2,e_1,e_2)\derives\bet$.
In order to achieve this, it suffices to introduce the production rules
\begin{align*}
	\bridge(p_1,p_2,q_1,q_2\cdot \ov a) &\production a\, \bridge(p_1\cdot a, p_2,q_1,q_2), \\
	\bridge(p_1,p_2,p_1,p_2) &\production 1
\end{align*}
for $p_1,q_1\in\cQ_1$, $p_2,q_2\in\cQ_2$, and $a\in\Sig$.
Observe that every derivation from $\bridge(p_1,p_2,q_1,q_2)$
to a terminal word must use the rule $\bridge(q_1,p_2, q_1,p_2) \production 1$
as last step.
Thus, $\bridge(p_1,p_2,q_1,q_2) \derives w$ implies $p_1\cdot w = q_1$ and
$p_2\cdot \ov{w} = q_2$ as desired.
Furthermore, for all words $w \in \bridge(p_1,p_2,q_1,q_2)$ and
all factorizations $w = uv$
(\ie $(p_1\cdot u)\cdot v = q_1$ and $(p_2\cdot \ov v)\cdot\ov u = q_2$)
there is a derivation
\begin{equation*}
	\bridge(p_1,p_2,q_1,q_2) \derives u\, \bridge(p_1\cdot u,p_2,q_2,p_2\cdot \ov v)
		\derives uv\, B(q_1, p_2,q_1,p_2) \derive uv
\end{equation*}
where the non-terminal reached after $\abs u$ steps is determined.
We conclude, the non-terminal $\bridge(p_1,p_2,q_1,q_2)$ derives all words from the language
$\bridge(p_1,p_2,q_1,q_2)$ and the derivation of each word is unambiguous.

The linear context-free part of the grammar $G$ are the derivations of the minimal
gamma-alpha-prefixes and the corresponding suffixes.
In a similar manner as above,
for every quadruple $(p_1,p_2,q_1,q_2)\in \cQ_1 \times \cQ_2\times \cQ_1\times \cQ_2$
we let $\mgp(p_1,p_2,q_1,q_2)$ be a non-terminal in $G$
and for $p_1,q_1\in\cQ_1$, $p_2,q_2\in\cQ_2$, and $a\in \Sig$ we 
define a rule
\begin{equation*}
	\mgp(p_1,p_2,q_1\cdot\ov a,q_2\cdot\ov a) \production a\, \mgp(p_1\cdot a,p_2\cdot a,q_1,q_2)\,\ov a
\end{equation*}
if $q_1\cdot\ov a\notin\cF_1$ and $q_2\cdot\ov a\notin\cF_2$;
and for $p_1,q_1\in\cQ_1$, $p_2,q_2\in\cQ_2$, and $\alp\in\Sig^\kap$ we define a rule
\begin{equation*}
	\mgp(p_1,p_2,q_1\cdot\ov \alp,q_2\cdot\ov \alp) \production
		\alp \, \bridge(p_1\cdot \alp,p_2\cdot \alp,q_1,q_2)\,\ov\alp
\end{equation*}
if $q_1\cdot\ov \alp\in\cF_1$ or $q_2\cdot\ov \alp\in\cF_2$.
Observe that the derivations we introduce are again unambiguous
since on a derivation
\begin{equation*}
	\mgp(q_{01},q_{02},q_1',q_2') \derives u\, \mgp(p_1,p_2,q_1,q_2)\,\ov u
	\derives uv \, \mgp(c_1,c_2,f_1,f_2)\, \ov v\ov u
\end{equation*}
the non-terminal $\mgp(p_1,p_2,q_1,q_2)$ is determined by $p_i = q_{0i}\cdot u$
and $q_i = f_i\cdot \ov v$ (for $i =1,2$).
Furthermore, the states $q_1'$, $q_2'$, $q_1$, and $q_2$ cannot be final states
(if $v\neq \e$) and if $f_1\in\cF_1$ or $f_2\in\cF_2$, then we have to use
a production rule of the second form in the next derivation step.

We conclude, there is a situation as in \reffig{fig:urun}
if and only if
\begin{align*}
	\mgp(q_{01},q_{02},q_1',q_2')
		&\derives \gam\, \mgp(c_1,c_2,f_1,f_2)\, \ov \gam \\
		&\derive \gam\alp\, \bridge(d_1,d_2,e_1,e_2)\,\ov\alp\ov\gam \\
		&\derives \gabag = \pi
\end{align*}
and the derivation of $\pi$ is unambiguous.
Thus, we let $\mgp(q_{01},q_{02},q_1',q_2')$ be the axioms in the grammar $G$
for all $q_1'\in\cQ_1$ and $q_2'\in\cQ_2$.
Since for each word $\pi$ there exists at most one axiom with 
$\mgp(q_{01},q_{02},q_1',q_2')$ such that 
$\mgp(q_{01},q_{02},q_1',q_2') \derives \pi$
(namely, $q_1' = q_{01}\cdot \pi$ and $q_2' = q_{02}\cdot \ov\pi$),
we see that $G$ is unambiguous linear.

As for the size of the grammar, observe that the number of non-terminals is
bounded by $2n_1^2n_2^2$ and the number of production rules is bounded  by
\begin{equation*}
	\left(\abs{\Sig}^\kap + 2\abs\Sig + \frac{1}{n_1n_2}\right)n_1^2n_2^2 \in \Oh(n^4).
\end{equation*}
\end{proof}

%%%%%%%%%%%%%%%%%%%%%%%%%%%%%%%%%%%%%%%%%%%%%%%%%%%%%%%%%%%%%%%%%%%%%%%%%%%%%%%%
%%%%%%%%%%%%%%%%%%%%%%%%%%%%%%%%%%%%%%%%%%%%%%%%%%%%%%%%%%%%%%%%%%%%%%%%%%%%%%%%
\section{Polynomial Time Decision Algorithm}\label{sec:algorithm}

We consider the following decision problem: 

\begin{description}
	\item[Input:]
		DFAs $\cA_1$ and $\cA_2$ (with state sets $\cQ_1$ and $\cQ_2$) accepting the languages $L_1$ and $\ov{L_2}$, respectively.

The input size is $n= \max\oneset{\abs{\cQ_1},\, \abs{\cQ_2}}$.
	\item[Question:]
		Is the hairpin completion \ccH regular?
\end{description}

The purpose of this section is to prove the following theorem.

\begin{theorem}\label{thm:main}
	The problem whether the hairpin completion \ccH is regular is
	decidable in time
	\begin{enumerate}[\quad i.)]
		\item $\Oh(n^2)$ if $L_1 = \es$ or $L_2 = \es$.
		\item $\Oh(n^6)$ if $L_1 = \ov{L_2}$.
		\item $\Oh(n^8)$ in general.
	\end{enumerate}
\end{theorem}

The algorithm deciding this problem is divided in Test~1, 2, and ~3.
Test~0 yields the time performance in case when $L_1=\es$ or $L_2=\es$,
yet it is redundant for the other cases.
The tests check properties of an automaton \cA which accepts
the minimal gamma-alpha-prefixes, introduced in \refsec{sec:unambiguity}.
We will start with the construction of $\cA$.

%%%%%%%%%%%%%%%%%%%%%%%%%%%%%%%%%%%%%%%%%%%%%%%%%%%%%%%%%%%%%%%%%%%%%%%%%%%%%%%%
\subsection{The Automaton $\cA$}\label{sec:automaton}

The non-deterministic automaton $\cA$, we are about to construct, will accept
those words that are a minimal gamma-alpha-prefix of some word $\pi = \gabag$
and the final states of the automatons will determine from which 
language $\bridge(d_1,d_2,e_1,e_2)$
we have to choose the factor $\bet$.
The construction is analogous to the definition of rules for
the non-terminals $R(p_1,p_2,q_1,q_2)$ in \refsec{sec:unambiguity}.

In order to improve the time bound in case when $L_1 = \ov{L_2}$,
we introduce the usual product automaton of $\cA_1$ and $\cA_2$ with state set
\begin{equation*}
	\cQ_{12}  = \set{(p_1,p_2)\in\cQ}{\exists w\in\Sig^* \colon q_{01}\cdot w = p_1 \land
		q_{02} \cdot w = p_2}
\end{equation*}
and operation $(p_1,p_2)\cdot w = (p_1\cdot w, p_2\cdot w)$
for $(p_1,p_2)\in\cQ_{12}$ and $w\in\Sig^*$.
Furthermore, we let $n_{12} = \abs{\cQ_{12}}$.
Note that if $L_2 = \es$ or $L_1 = \ov{L_2}$, then $n_{12} = n_1 = n$ and
in general $n \leq n_{12} \leq n^2$.
(Recall that $n_{i} = \abs{\cQ_{i}}$ for $i = 1,2$.)

Note first that a non-terminal $R(p_1,p_2,q_1,q_2)$ is reachable from an axiom
 only if $(p_1,p_2)\in\cQ_{12}$;
hence, we will consider states from
$\cQ_{12}\times\cQ_1\times\cQ_2\sse \cQ_1 \times \cQ_2\times \cQ_1\times \cQ_2$
for the construction of \cA.
From now on, we call $(p_1,p_2,q_1,q_2)$ 
a {\em basic bridge} if $\bridge(p_1,p_2,q_1,q_2) \neq \es$.
This notation is due to the fact, that there is some word that connects
the state pairs $(p_1,p_2)$ and $(q_1,q_2)$.
It is easy to see that in case when $(p_1,p_2,q_1,q_2)$ is not a basic bridge,
neither the non-terminal $\mgp(p_1,p_2,q_1,q_2)$
nor the non-terminal $\bridge(p_1,p_2,q_1,q_2)$
is productive in the grammar $G$.
In order to accept the $\alp$-factor, we also need
\emph{levels} for $0 \leq \ell \leq \kap$;
hence there are $\kap+1$ levels.
By $[\kap]$ we denote in this paper the set $\oneset{0, \ldots, \kap}$.
Define 
\begin{equation*}
	\set{((p_1,p_2),q_1,q_2, \ell) \in \cQ_{12} \times \cQ_{1} \times \cQ_{2} \times [\kap]}
		{(p_1,p_2,q_1,q_2) \text{ is a basic bridge}}
\end{equation*}
as the state space of $\cA$.
For $N = n_{12} n_1 n_2\leq n^4$ the size of \cA is bounded by
$N\cdot (\kap+1)\in \Oh(N)\sse\Oh(n^4)$. 
We have $N \leq n^2$ for $L_1 = \es$ or $L_2= \es$, and $N \leq n^3$ for $L_2 = \ov {L_1}$.

By a slight abuse of languages
we call  a state $((p_1, p_2) ,q_1,q_2,\ell)$ a \emph{bridge}.
Bridges are frequently denoted by $(P,q_1,q_2,\ell)$ with $P= (p_1, p_2)\in \cQ_{12}$, 
$q_1 \in \cQ_1$, $q_2\in \cQ_2$, and  $\ell \in [\kap]$.
Bridges are a  central concept in the following. 

The $a$-transitions in the NFA for  $a\in \Sigma$ are given by the following arcs: 
\begin{alignat*}{2}
	(P,\;q_1\cdot \ov a ,\; q_2 \cdot \ov a ,0) &\ras{a}{} (P\cdot a ,\;q_1,q_2,0)&\qquad&
		\text{ for } q_i\cdot \ov a \notin \cF_i, \, i = 1,2,\\
	(P,\;q_1\cdot \ov a ,\;q_2 \cdot  \ov a ,0) &\ras{a}{} (P\cdot a ,\;q_1,q_2,1)&&
		\text{ for } q_1\cdot \ov a \in \cF_1 \text{ or }q_2\cdot \ov a \in \cF_2, \\
	(P,\;q_1\cdot \ov  a ,\;q_2 \cdot  \ov a ,\ell) &\ras{a}{} (P\cdot a ,\;q_1,q_2,\ell + 1)&&
		\text{ for } 1 \leq \ell < \kap.
\end{alignat*}

Observe that  no state
of the form $(P,q_1,q_2,0)$ with $q_1 \in \cF_1$ or $q_2 \in \cF_2$ has an outgoing 
arc to level zero; we must switch to level one. There are no outgoing arcs on level~$\kap$, and 
for each  $(a, P,q_1,q_2,\ell) \in \Sigma \times \cQ_{12} \times \cQ_{1} \times \cQ_{2}\times [\kap-1] $
there exists at most  one arc $(P,q_1',q_2',\ell) \ras{a}{} (P\cdot a ,q_1,q_2,\ell')$.
Indeed, the triple $(q_1',q_2',\ell')$ is determined by $(q_1,q_2,\ell)$ and the letter $a$.
Not all arcs exist because
$(P,q_1',q_2',\ell)$ can be a bridge whereas  $(P\cdot a ,q_1,q_2,\ell')$ is not. 
Thus, there are at most $\abs{\Sig} \cdot N\cdot \kap\in\Oh(N)$ arcs in the NFA.

The set of initial states $\cI$ contains all bridges of the form
$(Q_0,q'_1,q'_2,0)$ where $Q_0 = (q_{01},\, q_{02})$. 
The set of final states $\cF$ is given by all bridges $(P,q_1,q_2,\kap)$ on level~$\kap$. 

For an example and a graphical presentation of the NFA, see \reffig{fig:automaton}.

\begin{figure}[ht]
  \centering
  \def\scale{.9}
  \begin{tikzpicture}[shorten >=1pt,node distance=2cm,auto,initial text=,%
  initial distance=4mm,bend angle=45,font=\footnotesize],scale=\scale]
    \tikzstyle{every node}=[scale=\scale]
    \tikzstyle{every loop}=[distance=.5cm]
	\node at (0,5.5)	[state,initial]		(A0)						{$q_{01}$};
	\node 			[state]				(A1)		[right of=A0]	{$p_1$};
	\node 			[state, accepting]	(A2)		[right of=A1]	{$f_1$};
	\node 			[state]				(A3)		[below of=A1]	{$\dead_1$};
	\node			[above of=A1,node distance=1.5cm]			{$L_1 = a^*(b\mid \ov b)\ov a$};
    
    \path [->]	(A0)	edge	[loop above]	node			{$a$}			()
    					edge					node			{$b,\ov b$}		(A1)
    					edge	[bend right]	node	[swap]	{$\ov a$}		(A3)
    			(A1)	edge					node			{$\ov a$}		(A2)
	    				edge					node			{$a,b,\ov b$}	(A3)
    			(A2)	edge	[bend left]		node			{$\Sigma$}		(A3)
    			(A3)	edge	[loop below]	node			{$\Sigma$}		();

    \node at (6,5.5)	[state,initial]		(B0)					{$q_{02}$};
    \node			[state]				(B1)	[right of=B0]	{$p_2$};
    \node			[state, accepting]	(B2)	[right of=B1]	{$f_2$};
    \node			[state]				(B3)	[below of=B1]	{$\dead_2$};
	\node			[above of=B1,node distance=1.5cm]			{$\ov{L_2} = a^*\ov b\ov a$};
    
    \path [->]	(B0)	edge	[loop above]	node			{$a$}			()
    					edge					node			{$\ov b$}		(B1)
    					edge	[bend right]	node	[swap]	{$\ov a,b$} 	(B3)
    			(B1)	edge					node			{$\ov a$}		(B2)
    					edge					node			{$a,b,\ov b$}	(B3)
    			(B2)	edge	[bend left]		node			{$\Sigma$}		(B3)
    			(B3)	edge	[loop below]	node			{$\Sigma$}		();

	\begin{scope}[scale=\scale]

    \tikzstyle{every state}=[rectangle,minimum height=.8cm,minimum width=2.3cm]
    \tikzstyle{every pin}=[pin distance=4mm]
    \tikzstyle{every pin edge}=[shorten <=1pt]
    \tikzstyle{init}=[pin={[pin edge={<-}]170:}]
    
    \node at (0,0)		[state,initial]		(A)		{$(Q_0,\dead_1,\dead_2,0)$};
    \node at (3,0)		[state,init]		(B)		{$(Q_0,f_1,f_2,0)$};
    \node at (6,0)		[state,accepting]	(B1)	{$(Q_0,p_1,p_2,1)$};
    \node at (7.25,0)	[right]						{$B(q_{01},q_{02},p_1,p_2)= b$};
    \node at (3,1)		[state,initial]		(C)		{$(Q_0,f_1,\dead_2,0)$};
    \node at (6,1)		[state,accepting]	(C1)	{$(Q_0,p_1,\dead_2,1)$};
    \node at (7.25,1)	[right]						{$B(q_{01},q_{02},p_1,\dead_2)= aa^+b \mid a^*\ov b$};
    \node at (6,2)		[state,accepting]	(C2)	{$(Q_0,p_1,f_2,1)$};
    \node at (7.25,2)	[right]						{$B(q_{01},q_{02},p_1,f_2)= ab$};
    \node at (3,-1)	[state,initial]		(D)		{$(Q_0,\dead_1,f_2,0)$};
    \node at (6,-1)	[state,accepting]	(D1)	{$(Q_0,\dead_1,p_2,1)$};
    \node at (7.25,-1)	[right]						{$B(q_{01},q_{02},\dead_1,p_2)= b\ov a\ov a^+$};
    \node at (6,-2)	[state,accepting]	(D2)	{$(Q_0,f_1,p_2,1)$};
    \node at (7.25,-2)	[right]						{$B(q_{01},q_{02},f_1,p_2)= b\ov a$};
    \node at (-.8,2) 	{$\cA$:};

	\path [->]	(A)	edge										node			{$a$}	(B)
					edge	[loop,out=125,in=55,distance=1cm]	node		 	{$a$}	(A)
					edge										node			{$a$}	(C)
					edge										node	[swap]	{$a$}	(D)
				(B)	edge										node			{$a$}	(B1)
				(C)	edge										node			{$a$}	(C1)
					edge										node			{$a$}	(C2)
				(D)	edge										node			{$a$}	(D1)
					edge										node	[swap]	{$a$}	(D2);
					
	\end{scope}
  \end{tikzpicture}
  \caption{DFAs for $L_1$ and $\ovc{L_2}$ and the resulting NFA $\cA$ with 4 initial states
  and 5 final states associated to the (linear context-free) \hpc
  $\ccH = a^+ b \ovc a ^+ \cup \{a^i \ovc b \ovc a^j\mid i \geq j \geq 1\}$ with $\kap=1$. }
  \label{fig:automaton}
\end{figure}

Next, we show that the automaton $\cA$ encodes the minimal gamma-alpha-prefixes
and that we obtain the hairpin completion \ccH in a natural way from \cA.
For languages $\bridge$ and $\mgp$ we denote by $\bridge^\mgp$ the language
\begin{equation*}
	\bridge^\mgp = \set{v\bet\ov v}{\bet\in \bridge \land v\in \mgp}.
\end{equation*}
(This notation is adopted from group theory where exponentiation denotes conjugation 
and the canonical involution refers to taking inverses.) 
Clearly, if $\bridge$ and $\mgp$ are regular, then $\bridge^\mgp$ is linear context-free,
but not regular in general.
Also note that if $\mgp$ is finite, then $\bridge^\mgp$ is regular.

\begin{lemma}\label{lem:str}
	Let $M = \cI \times \cF$.
	For each pair $\mu= (I,F)\in M$ with $F = ((d_1,d_2),e_1,e_2,\kap)$
	let $\mgp_\mu$ be the (regular) set of words which label a path from the initial state
	$I$ to the final state $F$, and let $\bridge_\mu = \bridge(d_1,d_2,e_1,e_2)$.

	The  \hpc $\ccH$ is the disjoint union
	\begin{equation*}
		\ccH= \bigcup_{\mu \in M}\CP.
	\end{equation*}
	
	Moreover, for $\mu\in\cI\times\cF$ and
	for all words $\bet\in\bridge_\mu$ and $v\in\mgp_\mu$,
	the minimal gamma-alpha-prefix of $v\bet\ov v$ is $v$.
\end{lemma}

\begin{proof}
Let $\pi\in\ccH$.
Let $\gam\alp$ be the minimal gamma-alpha-prefix of $\pi$ with $\abs\alp = \kap$
and factorize $\pi=\gabag$.
There are runs in the DFAs
\begin{align*}
	\cA_1: \quad& q_{01} \ras{\gam} c_1 \ras{\alp} d_1 \ras {\bet }{}
		e_1 \ras {\ov \alp}{} f_1 \RAS {\ov \gam }{} q_1', \\
	\cA_2: \quad& q_{02} \ras{\gam} c_2 \ras{\alp} d_2\ras {\ov \bet }{}
		e_2 \ras {\ov \alp}{} f_2 \RAS {\ov \gam}{} q_2'
\end{align*}
where $f_1\in\cF_1$ or $f_2\in\cF_2$ (cf.~\reffig{fig:urun}).
Recall that all states on these paths are determined by $\pi$.

By the definition of the NFA $\cA$,
we find a path $I\ras \gam A\ras\alp F$ where $I = (Q_0,q_1',q_2',0)$,
$A = ((c_1,c_2),f_1,f_2,0)$, and $F = ((d_1,d_2),e_1,e_2,\kap)$.
As $\bet\in\bridge(d_1,d_2,e_1,e_2)$,
there is a unique $\mu = (I,F) \in \cI \times \cF$ with 
$\pi \in \CP$.

Conversely, let $\mu=(I,F)\in\cI\times\cF$,
let $\bet\in \bridge_\mu$, and let $I \ras \gam A \ras \alp F$ with $\abs\alp = \kap$
be a path in $\cA$.
As $F$ is a final state it is on level~$\kap$ and $A = ((c_1,c_2),f_1,f_2,0)$
is the last state on level zero, whence $f_1\in\cF_1$ or $f_2\in\cF_2$.
Therefore, we find runs in the DFAs just like above where 
$I = (Q_0,q_1',q_2',0)$ and $F = ((d_1,d_2),e_1,e_2,\kap)$.
We conclude $\gam\alp$ is the minimal gamma prefix of $\gabag$ and $\gabag\in\ccH$.
\end{proof}

The next Lemma tells us that the paths in the automaton are unambiguous.
The arguments are essentially the same as 
used in \refsec{sec:unambiguity}.
The unambiguity of paths will become crucial later.

\begin{lemma}\label{lem:unam}
	Let $w \in \Sig^*$ be the label of a path in $\cA$ from a bridge 
	$A= (P,p_1,p_2,\ell) $ to $A' = (P',p_1',p_2',\ell')$,
	then the path is unique. This means that $B=B'$ whenever $w = uv$ and 
	\begin{align*}
		A &\ras{u}{} B %(P'',p'',q'',\ell'') 
			\ras{v}{} A', &&
		A \ras{u}{} B' %(P'',p'',q'',\ell'')
			\ras{v}{} A'.
	\end{align*}
\end{lemma}

\begin{proof}
It is enough to consider $u = a \in \Sig$. 
Let $B = (Q,q_1,q_2,m)$. Then we have 
$Q= P \cdot a$ and $q_i= p_i' \cdot \ov v$.
If $\ell= 0$ and  $p_i\notin \cF_i$ for $i = 1,2$, then 
$m= 0$, too; otherwise $m= \ell +1$. Thus, $B$ is determined by 
$A$, $A'$, and $u$, $v$. We conclude $B=B'$.
\end{proof}

For the decision algorithm we need to construct the automaton \cA
within the time bounds.
The automaton can be constructed in time $\Oh(n_1^2n_2^2)$;
that is $\Oh(n^2)$ in case $L_1=\es$ or $L_2=\es$ and $\Oh(n^4)$ otherwise.
Recall that the number of states and the number of transitions are in
$\Oh(N) \sse \Oh(n_1^2n_2^2)$ and
that the tuple
$(a, P,q_1,q_2,\ell) \in \Sigma \times \cQ_{12} \times \cQ_{1} \times \cQ_{2}\times [\kap-1]$
defines one transition
$(P,q_1\cdot\ov a,q_2\cdot \ov a,\ell) \ras{a}{} (P\cdot a ,q_1,q_2,\ell')$
(where $\ell'$ is determined by $\ell$, $q_1\cdot\ov a$, and $q_2\cdot \ov a$)
if and only if $(P\cdot a ,q_1,q_2,\ell')$ is a bridge.
Thus, it suffices to show that we can compute the set of bridges in time $\Oh(n_1^2n_2^2)$.

Furthermore, at this stage, we compute the set of all {\em $a$-bridges} for $a\in \Sig$,
where a basic bridge $\bridge(d_1,d_2,e_1,e_2)$ is called an $a$-bridge if
$\bridge(d_1,d_2,e_1,e_2)\cap a\Sig^* \neq \es$.
Later, we need the precomputed sets containing all $a$-bridges.

\begin{lemma}\label{lem:bridges}
	The set containing all basic bridges and 
	the sets containing all $a$-bridges for $a\in\Sig$, respectively,
	can be computed in time $\Oh(n_1^2n_2^2)$.
\end{lemma}

\begin{proof}
	Consider a transition system with state set $\cQ_1\times \cQ_2$ and
	transitions $(p_1,q_2) \ras a (q_1,p_2)$ for all $p_1\cdot a = q_1$ and
	$p_2 \cdot \ov a = q_2$
	(we use forward edges in $\cA_1$ and backwards edges in $\cA_2$).
	Note that there are $n_1n_2\cdot \abs\Sig$ transitions and 
	the transition system can be constructed in $\Oh(n_1n_2)$.
	
	There is a path $(p_1,q_2) \ras w (q_1,p_2)$ with $w\in\Sig^*$
	if and only if $p_1 \cdot w = q_1$ and $p_2\cdot \ov w = q_2$.
	Thus, a quadruple $(p_1,p_2,q_1,q_2)\in\cQ_1\times\cQ_2\times\cQ_1\times\cQ_2$ is
	a basic bridge if and only if a path from $(p_1,q_2)$ to $(q_1,p_2)$ exists
	and it is an $a$-bridge if and only if such a path exists that starts with
	an $a$-transition.
	
	In order to compute the sets of bridges, we run
	a depth-first reachability search for all triples $(p_1,q_2,a)\in \cQ_1\times\cQ_2\times\Sig$;
	for each pair $(q_1,p_2)\in\cQ_1\times\cQ_2$ that is reachable from $(p_1,q_2)$ by a path
	starting with an $a$-transition,
	we mark $(p_1,p_2,q_1,q_2)$ as basic bridge and as $a$-bridge.
	Since every depth-first search can be performed in $\Oh(n_1n_2)$,
	the whole computation can be done in $\Oh(n_1^2n_2^2)$.
\end{proof}

\begin{remark}\label{rem:reachable}
	For convenience, we will henceforth assume that all states in the automaton
	are reachable from an initial state and lead to some final state.
	Such an reachability test can easily be performed in $\Oh(n_1^2n_2^2)$;
	thus, this will not breach the time bounds.
\end{remark}

%%%%%%%%%%%%%%%%%%%%%%%%%%%%%%%%%%%%%%%%%%%%%%%%%%%%%%%%%%%%%%%%%%%%%%%%%%%%%%%%
\subsection{Test~0}

We consider the case when $L_1$ or $L_2$ is finite.
In this case we are able provide a simple necessary and sufficient condition for
the regularity of \ccH.

\begin{proposition}\label{prop:ciaa:finite} \ \\ \vspace{-\baselineskip}
\begin{enumerate}[\quad i.)]
	\item If the language $L(\cA)$ is finite, then \ccH
		is regular.
	\item If the language $L(\cA)$ is infinite and either $L_1$ is finite or $L_2$ is finite,
		then \ccH is not regular.
\end{enumerate}
\end{proposition}

\begin{proof}
Statement {\it i.)}\ follows directly by \reflem{lem:str}.

For {\it ii.)} let $L(\cA)$ be infinite.
There is a path
\begin{equation*}
	I \ras u A \ras v A\ras w F
\end{equation*}
in \cA where $I$ is an initial bridge, $F=((d_1,d_2),e_1,e_2)$ is a final bridge, and
$A\ras v A$ is a non-trivial loop
(by non-trivial we mean $v\neq \e$).
Note that $A$ is on level $0$ and hence $\abs w \geq k$.
Let $\alp$ be the suffix of $w$ of length $\kap$ and
let $\bet$ be a word from the language $\bridge(d_1,d_2,e_1,e_2)$.
We have $\pi_i = uv^iw\bet\ov w\ov v^i\ov u\in\ccH$ for all $i\geq0$.
Moreover, if a prefix of $\pi_i$ belongs to $L_1$, it is a prefix of $uv^iw\bet\ov \alp$
and if a suffix of $\pi_i$ belongs to $L_2$, it is a suffix of $\alp\bet\ov w\ov v^i\ov u$.

By contradiction, assume \ccH is regular and $L_1$ is finite.
Let $j\geq 1$ such that the power $v^j$ is idempotent in the syntactic monoid of \ccH,
hence
\begin{equation*}
	\pi = uv^{jk} w\bet \ov w \ov v^s\ov u \in\ccH
\end{equation*}
for $t\ge 1$.
We consider $t$ to be huge.
More precisely, we assume that $\pi$ is at least twice as long as the longest word in $L_1$ and
that $v^{jk}$ covers more than half of $\pi$.
The longest suffix of $\pi$ that belongs to $L_2$ is still a suffix of $\alp\bet \ov w \ov v^j\ov u$
which is far too short to build the hairpin;
hence a prefix from $L_1$ has to build the hairpin and it has to cover more than half of $\pi$
--- a contradiction.
By a symmetric argument $L_2$ is infinite, too.
\end{proof}

We check this property. Although, strictly speaking,
Test~0 is redundant for the general case.

\newcommand{\Tzero}{%
	\begin{test}{0}
		Decide whether or not $L(\cA)$ is finite. 
		If it is finite, then stop with the output that $\ccH$ is regular.
		If it is not finite but $L_1$ or $L_2$ is finite,
		then stop with the output that \ccH is not regular.
	\end{test}
}
\Tzero

In case when $L_1 = \es$ or $L_2 = \es$ the time complexity follows by the next lemma
as in these cases we can consider $n_1=1$ or $n_2=1$, respectively.

\begin{lemma}\label{lem:test:zero}
	Test~0 can be performed in time $\Oh(n_1^2n_2^2)$.
\end{lemma}

\begin{proof}
	Recall that every state in \cA is reachable and co-reachable,
	by \refrem{rem:reachable}.
	The language $L(\cA)$ is infinite if and only if \cA contains
	at least one non-trivial loop $A\ras v A$
	By the well-known algorithm of Tarjan \cite{Tarjan72} we
	can decompose a directed graph (as well as a finite automaton)
	into its strongly connected components
	in linear time with respect to the number of transitions.
	As the automaton \cA has $\Oh(n_1^2n_2^2)$ transitions,
	this yields the time complexity.
\end{proof}

%%%%%%%%%%%%%%%%%%%%%%%%%%%%%%%%%%%%%%%%%%%%%%%%%%%%%%%%%%%%%%%%%%%%%%%%%%%%%%%%
\subsection{Test~1}

By Test~0, we may assume in the following that $\cA$ accepts an infinite
language and that the set $S$
of non-trivial strongly connected components of the automaton $\cA$
has been computed.
Every non-trivial strongly connected component is 
on level $0$ and, moreover, as $\cA$ accepts an infinite
language, there is at least one.
For $s \in S$ let $N_s$ be the number of states in the component $s$.
Note that $\sum_{s\in S} N_s \leq N$.
By putting some linear order on the set of bridges, we assign to each $s\in S$
the least bridge $A_s$ and some shortest, non-empty word $v_s$ such that
$A_s \ras {v_s} A_s$.

The next lemma tells us that for a regular \hpc $\ccH$ every
strongly connected component $s\in S$ is a simple cycle,
and hence, the word $v_s$ is uniquely defined.

\begin{lemma}\label{lem:loop}
	Let the hairpin completion \ccH be regular, $s\in S$ be
	a strongly connected component, and $A_s \ras{w} F$ be
	a path from $A_s$ to a final bridge $F$.
	Then the word $w$ is a prefix of some word in $v_s^+$.

	In addition, the  word $v_s$ is uniquely defined and
	the loop $A_s\ras {v_s} A_s$ visits every other bridge
	$B\in s\sm\oneset A$ exactly once.
	Thus it forms a Hamiltonian cycle of $s$ and $\abs{v_s} = N_s$.
\end{lemma}

\begin{proof}
Let  $A = A_s$ and $v = v_s$.
Consider a path labeled by $w$ from $A$ to a final bridge $F=((d_1,d_2),e_1,e_2,k)$.
As all bridges are reachable, we find a word $u$ and an initial bridge $I$
such that
\begin{equation*}
	I\ras u A \ras v A \ras w F.
\end{equation*}
As the automaton $\cA$ accepts $uv^i w$ for all $i \geq 0$,
we see that $uv^i w \bet \ov w \ov v^i \ov u \in \ccH$ for all $i \geq 0$
and all $\bet \in  B(d_1,d_2,e_1,e_2)$.
As $\ccH$ is regular, there are $j \geq 1$ and $k > \abs{w \bet}$
such that $uv^{jk} w \bet \ov w \ov v^{j} \ov u \in \ccH$, by pumping.
Due to the definition of $\cA$, the longest suffix of $\pi$ belonging to
$L_2$ is a suffix of $\alp\bet\ov w \ov v^j\ov u$, where $\alp$ is the suffix
of $w$ of length $\kap$, and this suffix is too short to create the hairpin
completion.
This means that the \hpc is forced to use a prefix in $L_1$ and
that has to be a prefix of $uv^{jk}w\bet\ov \alp$.
Therefore, the suffix $\ov w \ov v^j\ov u$ is complementary to a prefix of $uv^{jk}$,
whence $w$ must be a prefix of  $v^{j(k-1)}$ (see~\reffig{fig:longhairpin})
and, thus, concludes the first statement of our lemma.

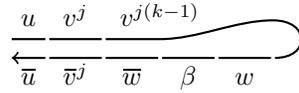
\begin{figure}[ht]
	\centering
	\begin{tikzpicture}[text height=1.5ex,text depth=.25ex,
			every path/.style={thick,shorten <=.75pt,shorten >=.75pt}]
		\draw  (-.5,0) -- node [above,sloped] {$u$} (0,0);
		\draw  (0,0) -- node [above,sloped] {$v^{j}$} (.75,0);
		\draw  (.75,0) -- node [above,sloped,at end] {$v^{j(k-1)}$} (1.5,0)
			.. controls +(right:.5) and +(left:.5) .. (3,.25)
			.. controls +(right:.5) and +(right:.5) .. (3,-.25);
		\draw  (3,-.25) -- node [below,sloped] {$w$} (2.25,-.25);
		\draw  (2.25,-.25) -- node [below,sloped] {$\bet$} (1.5,-.25);
		\draw  (1.5,-.25) -- node [below,sloped] {$\ov w$} (.75,-.25);
		\draw  (.75,-.25) -- node [below,sloped] {$\ov v^j$} (0,-.25);
		\draw [->] (0,-.25) -- node [below,sloped] {$\ov u$} (-.5,-.25);
%		\draw [-latex,gray,dashed] (0,.7) -- (2.75,.7);
%		\draw [latex-,gray,dashed] (0,-.9) -- (2.75,-.9);
	\end{tikzpicture}
	\caption{The hairpin of $\pi$
		(Read the upper part from left to right and the lower part from right to left).}
	\label{fig:longhairpin}
\end{figure}

Recall that $A\ras v A$ is a shortest, non-trivial loop around $A$;
hence $\abs{v}\leq N_s$ is obvious.
Let $B\in s\sm\oneset{A}$ and $x=x_1x_2$ such that $A\ras {x_1} B\ras {x_2} A$.
For some $i,j \geq 1$ we have $\abs{v^i} = \abs{x^j}$.
Thus, $v^i = x^j$ by the first statement.
By the unique-path-property  stated in \reflem{lem:unam}
we obtain that the loop $A\ras{x^j} A$ just uses the shortest loop $A \ras v A$
several times.
In particular, $B$ is on the shortest loop around $A$.
This yields $\abs v \geq N_s$ and hence the second statement.
\end{proof}

\begin{example}
	In the example given in \reffig{fig:automaton} the state $(Q_0,\dead_1,\dead_2,0)$
	forms the only strongly connected component and the corresponding path is labeled with $a$.
	As one can easily observe, the automaton $\cA$ satisfies the properties stated in
	\reflem{lem:loop}
	(even though the \hpc is not regular).
\end{example}

The next test tries to falsify the property of \reflem{lem:loop}. 
Hence it gives a sufficient condition that $\ccH$ is not regular. 

\newcommand{\Tone}{%
	\begin{test}{1}
		Decide whether there is $s\in S$
		and a path $A_s \ras w F$ such that $w$ is not a prefix of
		a word in $v_s^+$.
		If there is such a path, then stop with the output that \ccH is not regular.
	\end{test}
}
\Tone

\begin{lemma}
	Test~1 can be performed in time $\Oh(N^2)$.
\end{lemma}

\begin{proof}
For $s\in S$, let $A = A_s$ and compute a shortest non-empty word $v$
such that $A \ras v A$.
If $\abs v \neq N_s$, stop with the output that \ccH is not regular.
Otherwise, assign to each bridge that is reachable from $A$ a subset of marks
from $\oneset{0,\ldots,N_s -1}$.
A mark $i$ is assigned to a bridge $B$ if $B$ is reachable from $A$ with
a word from $v^*v[1,i]$.
Test~1 yields that \ccH is not regular if and only if there is a bridge that
is marked by $i$ and that has an outgoing $a$-transition where $a\neq v[i+1]$.
The marking algorithm can be performed by a depth-first search
that runs in time $\Oh(N\cdot N_s)$.
Summing over all strongly connected components we deduce
a time complexity in
$\Oh\left(\sum_{s\in S}N\cdot N_s\right)\sse \Oh(N^2)$.
\end{proof}

%%%%%%%%%%%%%%%%%%%%%%%%%%%%%%%%%%%%%%%%%%%%%%%%%%%%%%%%%%%%%%%%%%%%%%%%%%%%%%%%
\subsection{Test 2 and 3}\label{sec:ciaa:twothree}

Henceforth, we assume that Test~1 was successful
(\ie Test~1 did not yield that \ccH is not regular).
We fix a strongly connected component $s\in S$ of $\cA$. 
We let $A= A_s= ((p_1,p_2), q_1,q_2,0)$, we let $v= v_s$,
and we assume $A\ras v A$ forms an Hamiltonian cycle in $s$.
By $u$ we denote some word leading from an initial bridge
$((q_{01},q_{02}), q_1',q_2',0)$ to $A$. 
(For the following test we do not need to know $u$ we just need to know it exists.)
The main idea  is to investigate runs through the DFAs $\cA_1$ and $\cA_2$ where
$k,\ell \geq n$ according to \reffig{fig:otto}.

\begin{figure}[ht]
	\begin{alignat*}{3}
		&L_1: &\quad&q_{01} \ras{u} p_1 \ras{v^k} p_1 \ras{xy} c_1 \ras{z} d_1
			\ras{\ov x} e_1 \RAS{\ov v^{n_1}} &&q_1 \RAS{\ov v^*} q_1 \RAS{\ov u} q_1' \\
		&\ov{L_2}: &&q_{02} \ras{u} p_2 \ras{v^\ell} p_2 \ras{x} c_2 \ras{\ov z} d_2
			\ras{\ov y\ov x} e_2 \RAS{\ov v^{n_2}} &&q_2 \RAS{\ov v^*} q_2 \RAS{\ov u} q_2'
	\end{alignat*}
	\caption{Runs through $\cA_1$ and $\cA_2$ based on the loop
	$A\xrightarrow{\mspace{5mu}v\mspace{5mu}} A$.}
	% work-around since \oversets do not work in captions
	% $A\ras{v}A$}
	\label{fig:otto}
\end{figure}

We investigate  the case when $uv^k xyz \ov x \ov v^\ell \ov u \in \ccH$ for all $k \geq \ell$
and where (by symmetry) this property is due to the longest prefix belonging to $L_1$.

The following lemma is rather technical.
However, the notations are chosen to fit exactly to \reffig{fig:otto}.

\begin{lemma}\label{lem:egil}
	Let $x, y, z \in \Sig^*$ be words and $(d_1,d_2) \in \cQ_1 \times \cQ_2$ with
	the following properties:
	\begin{enumerate}[\quad 1.)]
		\item $\kap \leq \abs x < \abs v + \kap$ and $x$ is a prefix of some word in $v^+$.
		\item $0 \leq \abs y < \abs v$  and $xy$ is the longest common
			prefix of $xyz$ and some word in $v^+$.
		\item $z\in B(c_1,c_2, d_1,d_2)$, where
			$c_1 = p_1 \cdot xy$ and $c_2 = p_2 \cdot x$.
		\item $q_1 = d_1 \cdot  \ov {x} \ov {v}^{n_1}$ and during the computation of
			$d_1 \cdot  \ov {x} \ov {v}^{n_1}$ we see after 
			exactly  $\kap$ steps  a final state in $\cF_1$ and then never again.
		\item $q_2 = d_2 \cdot \ov {y} \ov {x} \ov {v}^{n_2}$ and, let $e_2 = d_2 \cdot \ov {y}\ov{x}$,
			during the computation of $e_2 \cdot \ov{v}^{n_2}$ we do not see a final state in $\cF_2$.
	\end{enumerate} 
	If $\ccH$ is regular, then  there exists a factorization
	$xyz\ov x \ov v = \mu \del\bet \ov \del \ov \mu$
	where $\abs\del = \kap$ and $p_2 \cdot\mu\del\ov\bet\ov\del \in \cF_2$
	(which implies $\del\bet \ov \del \ov \mu \ov v^*\ov u \sse L_2$). 
\end{lemma}

\begin{proof}
The conditions say that 
$ uv^k xyz  \ov x \ov v^{\ell} \ov u \in \ccH$ for all $k \geq \ell \geq n$.
Moreover, by condition~4, the \hpc can be achieved with a prefix in $L_1$, 
and the longest  prefix of $uv^k xyz  \ov x \ov v^{\ell} \ov u$ belonging to 
$L_1$ is the prefix $uv^k xyz  \ov \alp$ where $\ov\alp$ is the prefix of
$\ov x$ of length $\kap$. 

If $\ccH$ is regular, then we have $ uv^{k} xyz  \ov x \ov v^{k+1} \ov u \in \ccH$, too,
as soon as $k$ is large enough, by a simple pumping argument.
For this  \hpc we must use a suffix belonging to $L_2$.
For $z = 1$, this follows from $\abs y <\abs v$.
For $z\neq 1$ we use $\abs y <\abs v$ and, in addition,
that $xya$ with $a =z[1]$ is not a prefix of $vx$ by condition~2.

By~5 the longest suffix of  $ uv^{k} xyz  \ov x \ov v^{k+1} \ov u$ belonging to 
$L_2$ is a suffix of $xyz  \ov x \ov v^{k+1} \ov u$. 
Thus, we can write 
\begin{equation*}
	uv^{k} xyz  \ov x \ov v^{k+1} \ov u = uv^{k} xyz  \ov x \ov v \ov v^{k} \ov u
	=  uv^{k}  \mu \del\bet \ov \del \ov \mu \ov v^{k} \ov u
\end{equation*}
where  $\del\bet \ov \del \ov \mu \ov {v}^{k} \ov u \in L_2$ and $\abs\del = \kap$.
We obtain $xyz\ov x\ov v = \mu\del\bet\ov\del\ov\mu$.
As $p_2 = q_{02} \cdot u$ and $p_2 = p_2 \cdot v$,
we conclude $p_2 \cdot \mu\del\ov\bet\ov\del\in \cF_2$ as desired.
(Recall that our second DFA $\cA_2$ accepts $\ov {L_2}$.)
\end{proof}

\begin{example}
	Let us take a look at \reffig{fig:automaton} again.
	Let $A = (Q_0,\dead_1,\dead_2,0)$, $v = a$ and $u =\e$.
	If we choose $x = a$, $y=1$, $z=\ov b$, and $(d_1,d_2) = (p_1,p_2)$ we can
	see that conditions~1 to~5 of \reflem{lem:egil} are satisfied but there
	is no factorization $a\ov b\ov a\ov a = \mu\del\bet\ov\del\ov\mu$
	with $\abs\del = \kap = 1$ such that
	$q_{02}\cdot \mu\del\ov\bet\ov\del \notin\cF_2$.
	Hence, the \hpc is not regular.
\end{example}

We perform Test~2 and~3 which, again, try to falsify the property given
by \reflem{lem:egil} for a regular hairpin completion.
The tests distinguish whether the word $z$ is empty or non-empty. 

\newcommand{\Ttwo}{%
	\begin{test}{2}
		Decide the existence of words $x, y \in \Sig^*$ and states
		$(d_1,d_2) \in \cQ_1 \times \cQ_2$ satisfying
	\begin{enumerate}[\quad 1.)]
		\item $k \leq \abs x < \abs v + \kap$ and $x$ is a prefix of some word in $v^+$,
		\item $0 \leq \abs y < \abs v$  and $xy$ is a prefix of some word in $v^+$,
		\item $d_1 = p_1 \cdot xy$ and $d_2 = p_2 \cdot x$,
		\item $q_1 = d_1 \cdot  \ov {x} \ov {v}^{n_1}$ and during the computation of
			$d_1 \cdot  \ov {x} \ov {v}^{n_1}$ we see after 
			exactly  $\kap$ steps  a final state in $\cF_1$ and then never again, and
		\item $q_2 = d_2 \cdot \ov {y} \ov {x} \ov {v}^{n_2}$ and, let $e_2 = d_2 \cdot \ov {y}\ov{x}$,
			during the computation of $e_2 \cdot \ov{v}^{n_2}$ we do not see a final state in $\cF_2$
	\end{enumerate} 
		but where for all factorizations $xy\ov x \ov v = \mu \del\bet \ov \del \ov \mu$ 
		with $\abs\del = \kap$ we have  
		$p_2\cdot \mu\del\ov\bet\ov\del\notin\cF_2$.
		If we find such a situation, then stop with the output that $\ccH$ is not regular. 
	\end{test}
}

\begin{test}{2}
	Decide the existence of words $x, y \in \Sig^*$ and states
	$(d_1,d_2) \in \cQ_1 \times \cQ_2$ satisfying
	conditions~1 to~5 of \reflem{lem:egil} with $z=1$,
	but where for all factorizations $xy\ov x \ov v = \mu \del\bet \ov \del \ov \mu$ 
	with $\abs\del = \kap$ we have  
	$p_2\cdot \mu\del\ov\bet\ov\del\notin\cF_2$.
	If we find such a situation, then stop with the output that $\ccH$ is not regular. 
\end{test}
  
\newcommand{\Tthree}{%
	\begin{test}{3}
		Decide the existence of words $x, y, z \in \Sig^*$ 
		with $z \neq \e$ and states
		$(d_1,d_2) \in \cQ_1 \times \cQ_2$ satisfying
	\begin{enumerate}[\quad 1.)]
		\item $k \leq \abs x < \abs v + \kap$ and $x$ is a prefix of some word in $v^+$,
		\item $0 \leq \abs y < \abs v$  and $xy$ is the longest common prefix
			of $xyz$ and some word in $v^+$,
		\item $z\in B(c_1,c_2, d_1,d_2)$, where
			$c_1 = p_1 \cdot xy$ and $c_2 = p_2 \cdot x$,
		\item $q_1 = d_1 \cdot  \ov {x} \ov {v}^{n_1}$ and during the computation of
			$d_1 \cdot  \ov {x} \ov {v}^{n_1}$ we see after 
			exactly  $\kap$ steps  a final state in $\cF_1$ and then never again, and
		\item $q_2 = d_2 \cdot \ov {y} \ov {x} \ov {v}^{n_2}$ and, let $e_2 = d_2 \cdot \ov {y}\ov{x}$,
			during the computation of $e_2 \cdot \ov{v}^{n_2}$ we do not see a final state in $\cF_2$
	\end{enumerate} 
		but where for all factorizations $xyz\ov x \ov v = \mu \del\bet \ov \del \ov \mu$
		with $\abs\del = \kap$  we have
		$p_2\cdot \mu\del\ov\bet\ov\del\notin\cF_2$.
		If we find such a situation, then stop with the output that $\ccH$ is not regular. 
	\end{test}
}

\begin{test}{3}
	Decide the existence of words $x, y, z \in \Sig^*$ 
	with $z \neq \e$ and states
	$(d_1,d_2) \in \cQ_1 \times \cQ_2$ satisfying
	conditions~1 to~5 of \reflem{lem:egil},
	but where for all factorizations $xyz\ov x \ov v = \mu \del\bet \ov \del \ov \mu$ 
	with $\abs\del = \kap$ we have
	$p_2\cdot \mu\del\ov\bet\ov\del\notin\cF_2$.
	If we find such a situation, then stop with the output that $\ccH$ is not regular. 
\end{test}

Before we analyze the time complexity of Test~2 an Test~3
we will prove that if languages $L_1$ and $L_2$ pass
the tests we described so far, then the hairpin completion \ccH
is regular.
Thus, the properties given by
\reflem{lem:loop} and~\reflem{lem:egil} together are sufficient for
the regularity of \ccH.
The time complexity analysis of Test~2 and Test~3 can be found in
\refsec{sec:complexity}.

\begin{lemma}\label{lem:finaltests}
	Suppose no outcome of Tests~1, Test~2, and Test~3 is that \ccH is not regular.
	Then the \hpc $\ccH$ is regular. 
\end{lemma}

\begin{proof}
Let  $\pi  \in \ccH$. Write 
$\pi =  \gam \alp  \bet \ov \alp \ov \gam$
such that $\gam\alp$ is the minimal gamma-alpha-prefix of $\pi$ and
$\abs\alp = \kap$. 
Therefore, either $\gam \alp \bet \ov \alp \in L_1$ or 
$\alp   \bet \ov \alp \ov \gam \in L_2$;
we assume $\gam \alp \bet \ov \alp \in L_1$, by symmetry.
In addition, we may assume that $\abs \gam> n^4$
(cf.~\refprop{prop:ciaa:finite} and Test~0).
We can factorize $\gam = uvw$ with $\abs {uv} \leq n^4 $ and $\abs v \geq 1$
such that there are runs as in \reffig{fig:karl}
where $f_1\in\cF_1$.

\begin{figure}[ht]
	\begin{alignat*}{3}
		&L_1: &\quad&q_{01} \ras{u}{} p_1 \ras {v}{} p_1 \ras{w\alp\beta\ov\alp}{}&&
			f_1 \RAS {\ov w}{} q_1 \RAS {\ov v}{} q_1 \RAS {\ov u }{} q_1' \\
		&\ov{L_2}: &&q_{02} \ras{u}{} p_2 \ras {v}{} p_2 \ras{w\alp\ov\beta\ov\alp}{}&&
			f_2 \RAS {\ov w}{} q_2 \RAS {\ov v}{} q_2 \RAS {\ov u }{} q_2'
	\end{alignat*}
	\caption{Runs through $\cA_1$ and $\cA_2$ for the word $\pi$.}
\label{fig:karl}
\end{figure}

We infer from Test~1 that $w \alp$ is a prefix of some word in $v^+$. 
Hence, we can write $w \alp \bet = v^i xyz$ with $i\geq 0$
such that  $v^i xy$ is the maximal common prefix of $w\alp \bet $ and some word in $v^+$,
$w \alp \in v^*x$ with $\kap \leq \abs x  < \abs v+\kap$, and $\abs{y}< \abs{v}$. 

We see that for some $k \geq \ell \geq 0$ we can write 
\begin{equation*}
	\pi = uv^k xyz  \ov x \ov v^{\ell} \ov u.
\end{equation*}

Moreover, $uv^k xyz  \ov x \ov v^{\ell} \ov u \in \ccH$ for all $k \geq \ell \geq 0$.
There are only finitely many choices for $u,v,x,y$ (due to the lengths bounds) 
and for each of them there is a regular set $R_z$
associated to the finite collection of bridges such that
\begin{equation*}
	\pi \in \set{uv^k xyR_z  \ov x \ov v^{\ell} \ov u}{k \geq \ell \geq 0} \sse \ccH.
\end{equation*}

More precisely, we can choose $R_z = \oneset{\e}$ for $z=\e$ and otherwise we can choose 
\begin{equation*}
	R_z \in \set{B(c_1,c_2,d_1,d_2) \cap a \Sig^*}{
	(c_1,c_2,d_1,d_2) \text{ is a bridge and }a \in \Sig}.
\end{equation*}

Note that the sets $\set{uv^k xyR_z  \ov x \ov v^{\ell} \ov u}{k \geq \ell \geq 0}$ 
are not regular in general. 
If we bound however $\ell$ by $n$, then the finite union 
\begin{equation*}
	\bigcup_{0 \leq \ell\leq n}\set{uv^k xyR_z  \ov x \ov v^{\ell} \ov u}{k \geq \ell}
\end{equation*}
is regular.
Thus, we may assume that $\ell > n$.
Let $e_2 = p_2 \cdot x\ov z\ov y \ov x$. We have $e_2 \cdot \ov v^{n} = q_2$
and if we see a final state during the computation of $e_2\cdot \ov v^{n}$, then
for all $\ell > k\geq n$ and $z\in R_z$ we see that $uv^k xyz  \ov x \ov v^{\ell} \ov u \in \ccH$,
due to a suffix in $L_2$ and
\begin{equation*}
	uv^n v^+ xyR_z  \ov x \ov v^{+} \ov v^{n} \ov u \sse \ccH.
\end{equation*}

Otherwise, Test~2 or Test~3 tells us that for all $z\in R_z$ the word
$xyz\ov x \ov v$ has a factorization $\mu\delta\nu\ov\delta\ov\mu$
such that $\abs\delta = \kap$ and $p_2\cdot \mu\del\ov\nu\ov\del\in\cF_2$.
The paths $q_{02} \cdot u = p_2$ and $p_2\cdot v = p_2$
yield $\delta\nu\ov\delta\ov\mu \ov v^*\ov u \sse L_2$
and, again,
\begin{equation*}
	uv^n v^+ xyR_z  \ov x \ov v^{+} \ov v^{n} \ov u \sse \ccH.
\end{equation*}

Hence, the \hpc $\ccH$ is a finite union of regular languages
and, therefore, regular itself.
\end{proof}

%%%%%%%%%%%%%%%%%%%%%%%%%%%%%%%%%%%%%%%%%%%%%%%%%%%%%%%%%%%%%%%%%%%%%%%%%%%%%%%%
\subsection{Time Complexity of Test~2 and Test 3}\label{sec:complexity}

In this section we provide the final step of the proof of \refthm{thm:main}.
We show that Test~2 can be performed in time $\Oh(N^2)$ and
that Test~3 can be performed in time $\Oh(n_{12}n_1^2n_2^2 n)$.
Thus, in case when $L_1 = \ov{L_2}$ both tests run in $\Oh(n^6)$ and
in general Test~2 runs in $\Oh(n^8)$ and Test~3 runs in $\Oh(n^7)$.

\Ttwo

\begin{lemma}
	Test~2 can be performed in time $\Oh(N^2)$.
\end{lemma}

\begin{proof}
For a strongly connected component $s\in S$ with $A_s = ((p_1,p_2),q_1,q_2)$
and $v_s = v$,
we have to compute all words $x$ and $y$ such that there are runs
\begin{align*}
	&  p_1\ras{xy}{} d_1 \ras{\ov x \ov v^{n_1}}{} q_1,
	&& p_2\ras{x}{} d_2 \ras{\ov y\ov x\ov v^{n_2}}{} q_2
\end{align*}
and the conditions~1 to~5 are satisfied.
In addition, we demand that 
during the computation of $d_2 \cdot \ov{y} \ov {x} \ov{v}^{n_2}$ 
we do not  meet any final state in $\cF_2$ after more than $\kap-1$ steps.
(In case such a final state exists, either condition 5 is breached or
a factorization $x y\ov x \ov v = \mu\delta\beta\ov\delta\ov\mu$ with $\abs\delta = \kap$ and
$p_2\cdot \mu\del\ov\bet\ov\del \in \cF_2$ exists.)
By backwards searches in $\cA_1$
and $\cA_2$ starting at states $q_1$ and $q_2$, respectively,
and searching for paths labelled by suffixes of $\ov v^+$,
we compute all pairs $(x,xy)$ satisfying these conditions
in time $\Oh(N\cdot N_s)$.

At this stage we also compute the position $\ell(x,xy)$ of the last final state
during the run $p_2 \cdot vx\ov y \ov x$
and we let $\ell(x,xy) = 0$ if no such state exists.
Note that $0\leq \ell(x,xy) < N_s + \abs{x} + \kap$.
If a factorization $xy\ov x\ov v = \mu\del\bet\ov\del\ov\mu$ with $\abs\del = \kap$
and $p_2\cdot\mu\del\ov\bet\ov\del\in\cF_2$ exists, then $\abs{xy\ov x\ov v} - \ell(x,xy)$
gives us a lower bound for the length of $\mu$.

Let $m(x,xy)$ be the length of the longest $\mu$ such that a
factorization $xy\ov x \ov v = \mu \del \bet \ov \del \ov \mu$ with $\abs\del = \kap$ exists
(without the condition $p_2\cdot \mu\del\ov\bet\ov\del \in \cF_2$).

There is a factorization $xy\ov x\ov v = \mu \del \bet \ov \del \ov \mu$ with $\abs\del = \kap$
and $p_2\cdot \mu\del\ov\bet\ov\del \in \cF_2$ if and only if
$m(x,xy) \geq \abs {xy\ov x \ov v}-\ell(x,xy)$ and $\ell(x,xy) -\kap \geq \abs {xy\ov x \ov v}/2$.

We need to precompute the values $m(x,xy)$ efficiently, which turns out
to be a little bit tricky.
For $0\leq i< N_s$ we let $v_i = v[i+1,N_s]v[1,i]$
be the conjugate of $v$ starting at the $(i+1)$-st letter.
We wish to match position in $v_i^2$ with positions in $\ov v^2$.
For each $0\leq j< N_s$ we store the maximal $k\leq N_s$ such that
$v_i^2[j,j+k] = \ov v^2[j,j+k]$ in a table entry $M(i,j)$, see \reffig{fig:match}.
For each $i$ one run (from right to left) over the words $v_i^2$ and $\ov v^2$ is enough.
It takes $\Oh(N_s^2)$ time to build the table $M$.
Now, if we know the length $m'$ of the longest common prefix
of $v_{\abs{xy}}$ and $\ov x \ov v$,
then $m(x,xy) = \abs{xy}+m'-\kap$ (yet at most $\abs{xy\ov x\ov v}/2-\kap$).
The length of $m'$ is stored in
$M(\abs{xy\ov x}\bmod N_s,(-\abs{\ov x})\bmod N_s)$,
hence we have access to $m(x,xy)$ in constant time.

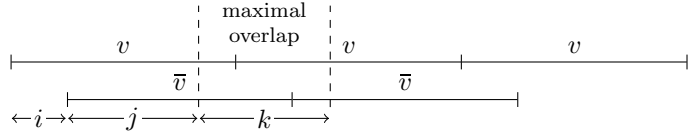
\begin{figure}[ht]
	\centering
	\begin{tikzpicture}
		\draw [|-|] (0,0) -- node [above] {$v$} (3,0);
		\draw [-|] (3,0) -- node [above] {$v$} (6,0);
		\draw [-|] (6,0) -- node [above] {$v$} (9,0);
		
		\draw [|-|] (.75,-.5) -- node [above] {$\ov v$} (3.75,-.5);
		\draw [-|] (3.75,-.5) -- node [above] {$\ov v$} (6.75,-.5);
		
		\draw [<->,shorten <=.5pt,shorten >=.5pt] (0,-.75) --
			node [inner sep=1pt,fill=white] {$i$} (.75,-.75);
		\draw [<->,shorten <=.5pt,shorten >=.5pt] (.75,-.75) --
			node [inner sep=1pt,fill=white] {$j$} (2.5,-.75);
		\draw [<->,shorten <=.5pt,shorten >=.5pt] (4.25,-.75) --
			node [inner sep=1pt,fill=white] {$k$} (2.5,-.75);
			
		\draw [dashed] (2.5,-.6) -- (2.5,.75);
		\draw [dashed] (4.25,-.6) -- (4.25,.75);
		\node at (3.375,.5) [text centered,text width=1.75cm, font=\footnotesize]
			{maximal overlap};
		
	\end{tikzpicture}
	\caption{Matching positions of $v_i^2$ with $\ovc v^2$.}
	\label{fig:match}
\end{figure}

All in all Test~4 can be performed in
$\Oh(\sum_{s\in S}N\cdot N_s)\sse \Oh(N^2)$.
\end{proof}

\Tthree

\begin{lemma}
	Test~3 can be performed in time $\Oh(n_{12}n_1^2n_2^2n)$.
\end{lemma}

\begin{proof}
For $s\in S$ with $A_s = ((p_1,p_2),q_1,q_2)$ and $v_s = v$,
we create two tables $T_1$ and $T_2$.
The table $T_1$ holds all pairs $(c_2,d_1)\in \cQ_2\times\cQ_1$ such that a word $x$
exists with
\begin{enumerate}[\quad 1.)]
  \item $\kap\leq \abs x < \abs v +\kap$ and $x$ is a prefix of a word in $v^+$,
  \item $p_2 \cdot x = c_2$,
  \item $d_1 \cdot \ov {x} \ov v^{n_1} = q_1$, and during the
    computation of $d_1 \cdot \ov {x} \ov v^{n_1}$ we see a final state after exactly $\kap$ steps and
    then never again.
\end{enumerate}
We call $x$ a witness for $(c_2,d_1)\in T_1$.
The table $T_2$ holds all triples $(c_1,d_2,a)\in\cQ_1\times\cQ_2\times\Sigma$ such that a
proper prefix $y' < v$ exists with 
\begin{enumerate}[\quad 1.)]
  \item $y'a$ is no prefix of $v$,
  \item $p_1\cdot y' = c_1$,
  \item $d_2 \cdot \ov{y'} \ov v^{n_2} = q_2$,
    and during the computation of $d_2\cdot \ov{y'} \ov v^{n_2}$ we do not see a final state
    after $\kap$ or more steps.
\end{enumerate}
We call $y'$ a witness for $(c_1,d_2,a)\in T_2$.
By backwards computing in the second component,
the tables $T_1$ and $T_2$ can be created in $\Oh(N_s n_1)$
and $\Oh(N_s n_2)$, respectively.

We claim that Test~3 yields that \ccH is not regular if and only if
there exists a pair $(c_2,d_1)\in T_1$ and a triple
$(c_1,d_2,a)\in T_2$ such that $(c_1,c_2,d_1,d_2)$ is an $a$-bridge.
Recall that the list of $a$-bridges is precomputed.

First, assume $(c_2,d_1)\in T_1$, $(c_1,d_2,a)\in T_2$, and $(c_1,c_2,d_1,d_2)$ is
indeed an $a$-bridge.
Let $x$ and $y'$ be the the witnesses for $(c_2,d_1)\in T_1$ and $(c_1,d_2,a)\in T_2$, respectively.
Choose $z\in B(c_1,c_2,d_1,d_2)\cap a\Sigma^*$ and $y$ such that $xy$ is a prefix of some word in $v^+$,
$\abs {xy} \equiv \abs{y'}\pmod{\abs v}$, and $\abs y < \abs v$.
Verify that $x,y,z$ and $(d_1,d_2)$ satisfy the conditions 1 to 5 of Test~3.
However, for any factorization
$xyz\ov x\ov v = \mu\delta\beta\ov \delta\ov\mu$ with $\abs\delta = \kap$,
the word $\mu\delta$ has to be a prefix of $xy$, since $xya$ is no prefix of $vx$.
During the computation of 
$d_2\cdot \ov {y'}\ov v^{n_2}$ we did not see a final state after more than $\kap-1$ steps.
The same holds for the computation of $d_2\cdot \ov y\ov x\ov v^{n_2}$ and,
therefore, we have $p_2 \cdot \mu\del\ov\bet\ov\del\notin\cF_2$.
  
Now assume that $x,y,z\in\Sigma^*$, $z\neq\e$, and $(d_1,d_2)\in\cQ_1\times\cQ_2$ exist,
which satisfy the conditions~1 to~5 of Test~3 but where for all factorizations
$xyz\ov x\ov v = \mu\delta\beta\ov \delta\ov\mu$ with $\abs\del = \kap$ we have
$p_2 \cdot \mu\del\ov\bet\ov\del\notin\cF_2$.
Choose $y' < v$ such that $\abs {xy} \equiv \abs{y'}\pmod{\abs v}$.
Let $c_2 = p_2\cdot x$, $c_1 = p_1 \cdot y'$
and $a\in\Sigma$ be the first letter of $z$.
Obviously, $(c_1,c_2,d_1,d_2)$ is an $a$-bridge
and $x$ is a witness for $(c_2,d_1)\in T_1$.
If we saw a final state after more than $\kap-1$ steps during the computation of
$d_2\cdot \ov {y'} \ov v^{n_2}$, then a factorization
$xyz\ov x\ov v = \mu\delta\beta\ov \delta\ov\mu$ where $\abs\del = \kap$ and
$p_2 \cdot \mu\del\ov\bet\ov\del\in\cF_2$ would exist.
Thus, $y'$ is a witness for $(c_1,d_2,a)\in T_2$.

Since the table of $a$-bridges is precomputed (see \reflem{lem:bridges}),
this test can be performed in time $\Oh(\abs{T_1}\cdot\abs{T_2})$.
The set of all first components of $T_1$ (\resp~$T_2$) is bounded by both,
the size $N_s$ and $n_2$ (\resp~$n_1$).
Therefore, we have $\abs{T_1} \in\Oh( n_1 \cdot \min(N_s,n_2))$
and $\abs{T_2} \in\Oh(n_2 \cdot \min(N_s,n_1))$.
By symmetry, assume $n_2 \leq n_1$.

Test~3 can be performed in time
\begin{align*}
  &\Oh\Biggl(\sum_{s\in S}\bigl(
    N_s n_1 + N_s n_2 + n_1n_2 \cdot\min(N_s,n_1)\cdot \min(N_s,n_2)
    \bigr) \Biggr) \sse \\
  &\Oh\Biggl(n_{12}n_1^2n_2 + n_{12}n_1n_2^2 + \sum_{s\in S, N_s\geq n_2} n_1^2n_2^2
    + \sum_{s\in S, N_s < n_2} N_s^2 n_1n_2 \Biggr)
\end{align*}
(Recall that $n_1 \leq n \leq n_{12}\leq n_1n_2\leq n^2$
and $\sum_{s\in S} N_s \leq N = n_{12}n_1n_2$.)

Since there are at most $n_{12}n_1$ strongly connected components with a size of
$n_2$ or more states,
\begin{equation*}
	\sum_{s\in S, N_s\geq n_2} n_1^2n_2^2 \leq n_{12}n_1^3n_2^2.
\end{equation*}

For the last term we can use the approximation
\begin{equation*}
	\sum_{s\in S, N_s < n_2} N_s^2 n_1n_2 \leq
	\sum_{s\in S, N_s < n_2} N_s n_1n_2^2
	\leq n_{12}n_1^2n_2^3.
\end{equation*}

We conclude, Test~3 can be performed in time $\Oh(n_{12}n_1^2n_2^2n)$.
\end{proof}

%%%%%%%%%%%%%%%%%%%%%%%%%%%%%%%%%%%%%%%%%%%%%%%%%%%%%%%%%%%%%%%%%%%%%%%%%%%%%%%%
%%%%%%%%%%%%%%%%%%%%%%%%%%%%%%%%%%%%%%%%%%%%%%%%%%%%%%%%%%%%%%%%%%%%%%%%%%%%%%%%
\section{Rational Growth}\label{sec:growth}

Let $L_1'= L_1 \cap \bigcup_{\alp\in\Sig^\kap}\Sig^*\alp\Sig^*\ov\alp$ and
    $L_2'= L_2 \cap \bigcup_{\alp\in\Sig^\kap}\alp\Sig^*\ov\alp\Sig^*$.
Obviously,
$\Hk(L_1',L_2')= \ccH$. Thus, the growths of \ccH should be compared with the growths of  
$L_1'$ and $L_2'$ rather than with the growths of $L_1$ and $L_2$.
 The languages $L_1'$ and $L_2'$ are still regular 
and we can compute their growths. However, to simplify the notation, it is more 
convenient to assume from the very beginning that $L_1$ and $L_2$
contains only words that can form  hairpins. 
 Formally, we assume throughout this section that
\begin{align*}
	L_1 &\sse \bigcup_{\alp\in\Sig^\kap}\Sig^*\alp\Sig^*\ov\alp, &
	L_2 &\sse \bigcup_{\alp\in\Sig^\kap}\alp\Sig^*\ov\alp\Sig^*.
\end{align*}

Remember (\refsec{sec:grr}) that the \grr $\lam_L$ of a language $L$ 
says  that $\abs{L \cap \Sig^{m}}$ behaves essentially as $\lam_L^m$.

\begin{theorem}\label{thm:growthrate}
	Let $\lam = \max\oneset{\lam_{L_1}, \lam_{L_2}}$ be the maximum \grr
	of $L_1$ and $L_2$, and let   $\eta$ be the \grr of \ccH.
\begin{enumerate}[\quad i.)]
\item The value lies within
	\begin{equation*}
		\sqrt{\lam}\le \eta \le \lam.
	\end{equation*}
	In particular, the growth of \ccH is exponential (\resp polynomial, finite)
	if and only if the maximum growth of $L_1$ and $L_2$ is exponential
	(\resp polynomial, finite).
\item If \ccH is regular, then we have $\eta = \lam$.
	Thus, the \grr of \ccH is the maximum \grr of $L_1$ and $L_2$.
\end{enumerate}
 \end{theorem}

The theorem will follow by \reflem{lem:growth:lam} and~\reflem{lem:growth:eta} 
in \refsec{sec:proof:growth} which compare
the \grrs $\lam$ and $\eta$ with
the \grrs of the languages $\bridge_\mu$ and $\mgp_\mu$ for $\mu\in M$.
Before we can prove theses lemmas, we need some preliminary observations on
\grrs of (regular) languages.

%%%%%%%%%%%%%%%%%%%%%%%%%%%%%%%%%%%%%%%%%%%%%%%%%%%%%%%%%%%%%%%%%%%%%%%%%%%%%%%%
\subsection{Basic Facts about Growth Indicators}

Consider two languages $K_1$ and $K_2$.
It is well known that the \grr of their union is
$\lam_{K_1\cup K_2} =\max\oneset{\lam_{K_1},\lam_{K_2}}$.
Furthermore, if $K_1 \neq \es \neq K_2$ the \grr of their
concatenation is $\lam_{K_1 K_2} = \max\oneset{\lam_{K_1},\lam_{K_2}}$, too.

Now, let $K$ be a regular language.
The prefix closure of $K$ is defined as
\begin{equation*}
	\Pref(K) = \set{u\in\Sig^*}{\exists v\in\Sig^* \colon uv\in K}.
\end{equation*}
The next lemma shows that the \grrs of $K$ and its prefix closure coincide.
Note that this does not necessarily hold if $K$ is (unambiguous) linear.

\begin{lemma}\label{lem:prefix:closure}
	Let $K$ be a regular language, then $\lam_K = \lam_{\Pref(K)}$.
\end{lemma}

\begin{proof}
	As $K\sse \Pref(K)$, the inequation $\lam_K \le \lam_{\Pref(K)}$ is obvious.
	
	Conversely, let $k$ be a constant such that $K$ is accepted by a DFA of size $k$
	and let $m\in \N$.
	For a word $u\in \Pref(K)\cap \Sig^m$, there is some word $v$ such that $uv\in K$
	and, moreover, we may assume $\abs v \le k$.
	Let $h$ be a mapping $h\colon u \mapsto uv$ for $u\in \Pref(K)\cap \Sig^m$
	such that $uv\in K$ and $\abs v \le k$.
	Note that $h$ is injective (for a fixed $m$).
	Thus, we see that 
	\begin{equation*}
		\abs{\Pref(K)\cap \Sig^m} \le \sum_{i = m}^{m+k} \abs{K\cap \Sig^i}.
	\end{equation*}
	
	For all $\nu > \lam_K$ there exists $c$ such that 
	$\abs{K\cap \Sig^m} \le c \nu^m$ for all $m\in \N$.
	Therefore,
	\begin{equation*}
		\abs{\Pref(K)\cap \Sig^m} \le
		\sum_{i = m}^{m+k} c \nu^i \le
		c (k+1) \nu^k \nu^m.
	\end{equation*}
	
	We conclude $\lam_{\Pref(K)}\le \nu$ and as such $\lam_{\Pref(K)}  = \lam_K$.
\end{proof}

%%%%%%%%%%%%%%%%%%%%%%%%%%%%%%%%%%%%%%%%%%%%%%%%%%%%%%%%%%%%%%%%%%%%%%%%%%%%%%%%
\subsection{Proof of Theorem~\ref{thm:growthrate}}\label{sec:proof:growth}

Recall from \reflem{lem:str}, that
the hairpin completion is the disjoint union
\begin{equation*}
	\ccH = \bigcup_{\mu\in M} \CP.
\end{equation*}
We let $\brgr_\mu$ and $\mgpgr_\mu$ be the \grrs of $\bridge_\mu$ and $\mgp_\mu$, respectively.
By $\brgr = \max\set{\brgr_\mu}{\mu\in M}$ and $\mgpgr = \max\set{\mgpgr_\mu}{\mu\in M}$
we denote the maximum \grrs of all $\bridge_\mu$ and all $\mgp_\mu$, respectively.
The next lemma compares the \grr $\lam$ with the \grrs $\brgr$ and $\mgpgr$.

\begin{lemma}\label{lem:growth:lam}
	$\lam = \max\oneset{\brgr,\mgpgr}$.
\end{lemma}

\begin{proof}
	We start by proving $\lam \ge \max\oneset{\brgr,\mgpgr}$.
	Let $\mu\in M$ be fixed.
	For $\gam\alp\in \mgp_\mu$ with $\abs\alp = \kap$, and $\bet\in \bridge_\mu$
	either $\gaba\in L_1$ or $\abag\in L_2$.
	Thus, we may define a mapping
	$h\colon (\mgp_\mu \times \bridge_\mu) \to L_1\cup L_2$ such that
	\begin{equation*}
		h(\gam\alp,\bet) = \begin{cases}
			\gaba &\text{if } \gaba \in L_1 \\
			\abag &\text{otherwise}.
		\end{cases}
	\end{equation*}
	Obviously, $\abs{\gam\alp}+\abs{\bet} = \abs{h(\gam\alp,\bet)}-\kap$.
	Also note that a word $w\in L_1\cup L_2$ of length $m$ can form less than $2m$ hairpin completions.
	Therefore, the cardinality of the inverse image is $\abs{h^{-1}(w)} < 2m$.
	Using the mapping $h$, we can compare the growth $r_m = \abs{\mgp_\mu\bridge_\mu \cap \Sig^m}$
	with the growth $\ell_m = \abs{(L_1\cup L_2)\cap \Sig^m}$;
	that is $r_m \le 2(m+\kap)\cdot\ell_{m+\kap}$ for $m\in\N$.
	
	For $\nu > \lam =\lam_{L_1\cup L_2}$ we chose $\nu'$ from the open interval $(\lam,\nu)$.
	There exists $c'>0$ such that $r_m \le 2(m+\kap) c' \nu'^\kap \nu'^m$ for all $m\in \N$
	and, as the function $\nu^m$ growth faster than $\nu'^m$, there is some $c>0$ such that
	$r_m \le c\nu^m$ for all $m\in \N$.
	Therefore, $\max\oneset{\brgr_\mu,\mgpgr_\mu} \le \nu$ for all $\nu > \lam$, whence
	$\max\oneset{\brgr_\mu,\mgpgr_\mu} \le \lam$.
	As this inequation holds for all $\mu\in M$, we deduce $\lam \ge \max\oneset{\brgr,\mgpgr}$.
	
	Conversely, we will prove that $L_1$ is included in a language $K$ whose
	\grr is $\max\oneset{\brgr,\mgpgr}$.
	As there is a symmetric language that includes $L_2$, this yields
	$\lam \le \max\oneset{\brgr,\mgpgr}$.
	Let $B = \bigcup_{\mu\in M}B_\mu$ and $R = \bigcup_{\mu\in M}R_\mu$.
	We let $K$ be the prefix closure $K = \Pref(RB\Sig^\kap)$.
	As the \grr of $RB\Sig^\kap$ is $\lam_{RB\Sig^\kap} = \max\oneset{\brgr,\mgpgr}$
	and by \reflem{lem:prefix:closure}, we deduce
	$\lam_K = \max\oneset{\brgr,\mgpgr}$.
	
	Now, consider $w\in L_1$.
	By assumption, $w$ can form a hairpin on its right side.
	We let $\pi \in \Hk(\oneset w,\es)$ be a hairpin completion of $w$.
	Let $\gam\alp$ be the minimal gamma-alpha-prefix of $\pi$ with $\abs\alp = \kap$
	and $\bet$ such that $\pi = \gabag$.
	Note that $w$ has to be a prefix of $\gaba\in RB\Sig^\kap$
	(by the minimality of $\abs\gam$).
	Thus, we may conclude $L_1 \sse K$ as desired.
\end{proof}

Now, let us compare the \grr $\eta$ with the \grrs $\brgr$ and $\mgpgr$.

\begin{lemma}\label{lem:growth:eta}
	$\eta = \max\oneset{\brgr,\sqrt\mgpgr}$.
\end{lemma}

\begin{proof}
	Let $\tau_\mu$ be the \grr of $\CP$ for $\mu\in M$.
	Since $\ccH = \bigcup_{\mu\in M} \CP$, we see that $\eta = \max\set{\tau_\mu}{\mu\in M}$.
	Thus, in order to prove the claim, it suffices to show that
	$\tau_\mu = \max\oneset{\brgr_\mu,\sqrt{\mgpgr_\mu}}$ for $\mu\in M$.
	Let $\mu\in M$ be fixed from here on
	and recall that $\bridge_\mu$ and $\mgp_\mu$ are non-empty.
	We let 
	\begin{alignat*}{2}
		g_{\bridge_\mu}(z) &= \sum_{m\ge 0} b_m z^m
		&\qquad&\text{with }b_m = \abs{\bridge_\mu \cap \Sig^m}, \\
		g_{\mgp_\mu}(z) &= \sum_{m\ge 0} r_m z^m
		&\qquad&\text{with }r_m = \abs{\mgp_\mu \cap \Sig^m}.
	\end{alignat*}
	It will be convenient to let
	$r_{i+1\slash 2} = 0$ for $i\in\N$.

	First, let us prove $\tau_\mu\ge \sig_\mu$.
	Let $v\in \mgp_\mu$ and consider $K = v \bridge_\mu \ov v$.
	Obviously, $K \sse \CP$ and hence
	$\tau_\mu \ge \lam_K = \sig_\mu$.
	
	Next, we prove $\tau_\mu \ge \sqrt{\rho_\mu}$.
	Let $K = \oneset{\bet}^{\mgp_\mu}\sse \CP$ for some $\bet\in \bridge_\mu$.
	The generating function of $K$ is given as
	$g_K(z) = \sum_{m\ge 0} r_{(m-\abs{\bet})\slash 2} z^m$.
	For all $\nu > \lam_K$
	there exists $c> 0$ such that
	\begin{equation*}
		\forall m\in\N\colon r_{(m-\abs{\bet})\slash 2} \le c\nu^{m}
		\quad\iff\quad
		\forall m\in\N\colon r_m \le c\nu^{\abs\bet}(\nu^2)^m
	\end{equation*}
	and, therefore, $\nu^2 \ge \rho_\mu$.
	We conclude $\tau_\mu \ge \lam_K \ge \sqrt{\rho_\mu}$.
	
	Finally, we need to prove $\tau_\mu \le \max\oneset{\brgr_\mu,\sqrt{\mgpgr_\mu}}$.
	As $\CP$ is unambiguous, by \reflem{lem:str},
	\begin{equation*}
		g_{\CP}(z) = \sum_{m\ge 0} d_m z^m
		\qquad\text{with }d_m = \sum_{k+\ell = m} b_k r_{\ell\slash 2}.
	\end{equation*}
	For $\nu > \max\oneset{\brgr_\mu,\sqrt{\mgpgr_\mu}}$
	we choose $\nu'$ from the open interval
	$\left(\max\oneset{\brgr_\mu,\sqrt{\mgpgr_\mu}},\nu\right)$.
	By that choice, $\nu^m$ grows faster than $\nu'^m$ and
	there is $c' > 0$ such that for all $m\in\N$ and $k+\ell = m$, we have
	$ b_k r_{\ell\slash 2} \le c'\nu'^m$.
	Thus, there is $c > 0$ such that for all $m\in \N$, the inequality
	$d_m\le mc'\nu'^m \le c\nu^m$ holds.
	This deduces the last step in the proof,
	$\tau_\mu \le \max\oneset{\brgr_\mu,\sqrt{\mgpgr_\mu}}$.
\end{proof}

\reflem{lem:growth:lam} and \reflem{lem:growth:eta} yields a development
of the \grrs $\lam$ and $\eta$ as shown in \reffig{fig:growth}.

\begin{figure}[ht]
	\center
	
	\begin{tikzpicture}[text height=1.75ex,text depth=.25ex]
		\begin{scope}[xshift=-\textwidth/2]
			\draw [-latex,shorten <=-4pt] (0,0) -- 
				node [pos=.5,below] {$\mgpgr$}
				(3,0) node [below left] {$\brgr$};
			\draw [-latex,shorten <=-4pt] (0,0) --
				node [pos=.5,left] {$\mgpgr$}
				(0,3) node [below left] {$\lam$};
			\draw [dashed,thin,gray] (0,1.5) -- (2.75,1.5);
			\draw [dashed,thin,gray] (1.5,0) -- (1.5,2.75);
			\draw (0,1.5) -- (1.5,1.5) -- (2.75,2.75);
		\end{scope}

		\draw [-latex,shorten <=-4pt] (0,0) --
			node [pos=.28,below] {$\sqrt\mgpgr$}
			(3,0) node [below left] {$\brgr$};
		\draw [-latex,shorten <=-4pt] (0,0) --
			node [pos=.28,left] {$\sqrt\mgpgr$}
			(0,3) node [below left] {$\eta$};
		\draw [dashed,thin,gray] (0,3*.28) -- (2.75,3*.28);
		\draw [dashed,thin,gray] (3*.28,0) -- (3*.28,2.75);
		\draw (0,3*.28) -- (3*.28,3*.28) -- (2.75,2.75);
	\end{tikzpicture}
	
	\caption{\Grrs $\lam$ and $\eta$ in dependency of $\brgr$ and $\mgpgr$.}
	\label{fig:growth}
\end{figure}
It is easy to see that $\eta$ is at least $\sqrt\lam$ and at
most $\lam$ and, therefore, we deduce the first statement of \refthm{thm:growthrate}.
The second statement of \refthm{thm:growthrate}
claims that if the hairpin completion is regular, then $\lam = \eta$.
In case when \ccH is regular,
we infer from \reflem{lem:loop} that if the hairpin completion of \ccH
is regular, then the growth of all $\mgp_\mu$ is polynomial (more precisely, linear) or finite
(\ie $\mgpgr = 1$ or $\mgpgr = 0$).
We conclude
$\lam = \max\oneset{\brgr,\mgpgr} = \max\oneset{\brgr,\sqrt\mgpgr} = \eta$.

%%%%%%%%%%%%%%%%%%%%%%%%%%%%%%%%%%%%%%%%%%%%%%%%%%%%%%%%%%%%%%%%%%%%%%%%%%%%%%%%
%%%%%%%%%%%%%%%%%%%%%%%%%%%%%%%%%%%%%%%%%%%%%%%%%%%%%%%%%%%%%%%%%%%%%%%%%%%%%%%%
\section*{Final Remarks}

We proved that regularity of a hairpin completion of regular languages
is decidable in polynomial time.
Considering the two-sided hairpin completion of regular languages,
the decision algorithm, we presented, can be performed in time $\Oh(n^8)$
(\resp $\Oh(n^6)$ in case when $L_1 = \ov{L_2}$) which, at first, seems to be
a high degree for a polynomial time algorithm.
However, the first step of the algorithm is the construction of an automaton \cA which is 
already of size $\Oh(n^4)$ (\resp $\Oh(n^3)$).
Thus, when speaking of time complexity with respect to the size of \cA,
the algorithm uses quadratic time, only.
Furthermore, as we take into account all pairs of states of \cA, the time bound
seems optimal for this approach and further improvement of the time complexity
would probably call for a completely new approach.
For the one-sided hairpin completion of a regular language,
we provide a faster algorithm which runs in quadratic time.

The polynomial time bounds are due to the fact that we use DFAs 
for the specification of $L_1$ and $\ov{L_2}$. We do not know what happens if 
$L_1$ and $L_2$ are given by NFAs. We suspect that deciding 
regularity if \ccH might become $\PSPACE$-complete. But this 
has not been investigated yet.

By our second result, that the hairpin completion of regular languages is
always an unambiguous linear language, we are able to effectively compute the
growth function of the hairpin completion.
Moreover, we showed that the hairpin completion has an exponential growth
if and only if one of the underlying languages has an exponential growth
(given that every word from the underlying languages can form a hairpin).
More precisely, the \grr of the hairpin completion is at most as large
as the maximum \grr of the underlying languages and at least as large as its
square root.
In case when the hairpin completion is regular, we provided an even stronger
relationship between the \grrs.
In that case, the \grr of the hairpin completion coincides with the maximum
\grr of the underlying languages.
Our results about growths are trivial in case that $L_1$ and $L_2$ have polynomial growths.
However, the structure of regular languages with 
polynomial growths is well-understood \cite{Szilard1992} (\refsec{sec:grr}).
We believe that a study of \hpc{}s 
for this class of regular languages might lead to interesting results. 
We leave this to future research.

Another interesting problem concerns the hairpin lengthening
of regular languages, which is an operation familiar to the hairpin completion.
We call $\gam_1\aba\ov{\gam_2}$ a (right) hairpin lengthening
of $\gam_1\aba$ if $\gam_2$ is a suffix of $\gam_1$
and we call it a (left) hairpin lengthening of $\aba\ov{\gam_2}$ if $\ov{\gam_1}$
is a prefix of $\ov{\gam_2}$.
The hairpin lengthening $\cH\cL_\kap(L_1,L_2)$ of languages $L_1$ and $L_2$ is introduced
analogously to the hairpin completion.
It is known that the hairpin lengthening of regular languages is linear,
but in contrast to the \hpc it is not unambiguous, in general, see \cite{diss_kop}.
This might indicate that deciding 
regularity of the hairpin lengthening $\cH\cL_\kap(L_1,L_2)$
is more difficult than for the hairpin completion. To date it is not known 
whether regularity of $\cH\cL_\kap(L_1,L_2)$ is decidable.

%%%%%%%%%%%%%%%%%%%%%%%%%%%%%%%%%%%%%%%%%%%%%%%%%%%%%%%%%%%%%%%%%%%%%%%%%%%%%%%%
%%%%%%%%%%%%%%%%%%%%%%%%%%%%%%%%%%%%%%%%%%%%%%%%%%%%%%%%%%%%%%%%%%%%%%%%%%%%%%%%
%\bibliographystyle{abbrv}
%\bibliography{../../TRACES/traces}

\newcommand{\Ju}{Ju}\newcommand{\Ph}{Ph}\newcommand{\Th}{Th}\newcommand{\Ch}{C%
h}\newcommand{\Yu}{Yu}\newcommand{\Zh}{Zh}

\end{document}